\documentclass[letterpaper,twocolumn,twoside]{IEEEtran} 
\usepackage{amsmath, amssymb, bm, cite, epsfig, psfrag, amsthm}
\usepackage{algorithm,algorithmic}
\usepackage{bbm}
\usepackage[usenames]{color}
\usepackage{tikz}
\usepackage{etoolbox}
\usetikzlibrary{shapes,arrows}
\usepackage{hyperref}

\newtoggle{conference}
\togglefalse{conference} 

\usepackage[normalem]{ulem} 

\newcommand{\textb}[1]{\textcolor{black}{#1}}

\def\beq{\begin{equation}}
\def\eeq{\end{equation}}
\def\beqa{\begin{eqnarray}}
\def\eeqa{\end{eqnarray}}
\def\beqan{\begin{eqnarray*}}
\def\eeqan{\end{eqnarray*}}

\def\Z{{\mathbb{Z}}}

\def\R{{\mathbb{R}}}

\def\argmin{\mathop{\mathrm{arg\,min}}}

\def\Diag{\mathop{\mathrm{Diag}}}
\def\diag{\mathop{\mathrm{diag}}}
\newcommand{\gradient}{\boldsymbol{\nabla}} 
\newcommand{\Hessian}{\boldsymbol{\mathcal{H}}} 
\DeclareMathOperator{\prox}{prox} 
\DeclareMathOperator{\sgn}{sgn} 
\def\x{\times}

\newcommand*{\Cdot}{\raisebox{-0.5ex}{\scalebox{2.0}{$\cdot$}}}
\newcommand*\dif{\mathop{}\!\mathrm{d}} 

\newtheorem{definition}{Definition}
\newtheorem{theorem}{Theorem}

\newtheorem{assumption}{Assumption}

\def\bhat{\widehat{b}}

\def\arr{\rightarrow}
\def\Exp{\mathbb{E}}
\def\var{\mathop{\mathrm{var}}}
\def\tm1{t\! - \! 1}
\def\tp1{t\! + \! 1}

\newcommand{\zero}{\mathbf{0}}

\newcommand{\dbf}{\mathbf{d}}

\newcommand{\pbf}{\mathbf{p}}

\newcommand{\qbf}{\mathbf{q}}

\newcommand{\rbf}{\mathbf{r}}

\newcommand{\sbf}{\mathbf{s}}

\newcommand{\sbfhat}{\widehat{\mathbf{s}}}
\newcommand{\ubf}{\mathbf{u}}

\newcommand{\vbf}{\mathbf{v}}

\newcommand{\xbf}{\mathbf{x}}
\newcommand{\xbfhat}{\widehat{\mathbf{x}}}

\newcommand{\ybf}{\mathbf{y}}
\newcommand{\zbf}{\mathbf{z}}

\newcommand{\zbfhat}{\widehat{\mathbf{z}}}
\newcommand{\Abf}{\mathbf{A}}

\newcommand{\Ibf}{\mathbf{I}}

\newcommand{\Pbf}{\mathbf{P}}

\newcommand{\Qbf}{\mathbf{Q}}

\newcommand{\Sbf}{\mathbf{S}}

\def\mubf{{\boldsymbol \mu}}

\def\xibf{{\boldsymbol \xi}}
\def\taubf{{\boldsymbol \tau}}

\newcommand{\indic}[1]{\mathbbm{1}_{ \{ {#1} \} }}

\tikzstyle{block}=[rectangle,draw, fill=blue!20,
    minimum height=1cm, minimum width=1.5cm]
\tikzstyle{signal}=[coordinate,draw]

\title{Fixed Points of Generalized Approximate Message Passing
with Arbitrary Matrices}

\author{
Sundeep Rangan, 
Philip Schniter, 
Erwin Riegler, 
Alyson K. Fletcher,
Volkan Cevher
\thanks{S. Rangan (email: srangan@nyu.edu) is with
      the Department of Electrical and Computer Engineering,
      New York University, Brooklyn, NY.
      The work of S. Rangan was supported by the National Science Foundation under grants CCF-1116589 and IIP-1237821 as well as generous support from NYU WIRELESS affiliate memberships.}
\thanks{P.~Schniter (email: schniter@ece.osu.edu) is with
      the Department of Electrical and Computer Engineering,
      The Ohio State University, Columbus OH.
      The work of P. Schniter was supported by the National Science Foundation 
      under grants CCF-1018368, CCF-1218754, and CCF-1527162.}
\thanks{E.~Riegler (email: erwin.riegler@tuwien.ac.at) is with
    the Institute of Telecommunications,
    Technische Universit{\"a}t, Wien.}%
\thanks{A.~K. Fletcher (email: akfletcher@ucla.edu) is with
    the Departments of Statistics, Mathematics, and Electrical Engineering,
    University of California, Los Angeles.
    The work of A.K. Fletcher was supported by the NSF under grant CCF-1254204.}%
\thanks{V. Cevher (email: volkan.cevher@epfl.ch) is with
    Ecole Polytechnic, Lausanne, Switzerland}%
\thanks{This paper was presented in part at ISIT 2013 \cite{RanSRFC:13-ISIT}.}%
\markboth{GAMP Fixed Points}
    {Rangan et al.}
}

\begin{document}
\setlength{\arraycolsep}{0.8mm}

\maketitle
\begin{abstract}
The estimation of a random vector with independent components passed through
a linear transform followed by a componentwise (possibly nonlinear) output map
arises in a range of applications. Approximate message passing (AMP) methods,
based on Gaussian approximations
of loopy belief propagation, have recently attracted considerable attention
for such  problems. For large random transforms, these methods
exhibit fast convergence and admit precise analytic characterizations with
testable conditions for optimality, even for certain non-convex problem instances.
However, the behavior of AMP under general transforms is not fully understood.  In this paper,
we consider the Generalized AMP (GAMP) algorithm and relate the method to more common optimization techniques.  This analysis
enables a precise characterization of the GAMP algorithm fixed-points that applies to
arbitrary transforms.  In particular, we show that the fixed points
of the so-called max-sum GAMP algorithm for MAP estimation
are critical points of a constrained maximization of the posterior density.
The fixed-points of the sum-product GAMP algorithm for estimation of the posterior marginals
can be interpreted as critical points of a certain free energy.
\end{abstract}


\section{Introduction}
Consider the constrained optimization problem
\beq \label{eq:xzOpt}
    (\xbfhat,\zbfhat) := \argmin_{\xbf,\zbf} F(\xbf,\zbf) \quad
        \mbox{s.t. } \zbf=\Abf\xbf,
\eeq
where $\xbf \in \R^n$, $\zbf \in \R^m$, $\Abf \in \R^{m \x n}$,
and the objective function admits a decomposition of the form
\beqa
    \lefteqn{F(\xbf,\zbf) := f_x(\xbf) + f_z(\zbf) } \nonumber \\
    & & f_x(\xbf) = \sum_{j=1}^n f_{x_j}(x_j), \quad
        f_z(\zbf) = \sum_{i=1}^m f_{z_i}(z_i), \label{eq:optsep}
\eeqa
for scalar functions $f_{x_j}(\cdot)$ and $f_{z_i}(\cdot)$.
One example where this optimization arises is the estimation problem
in Fig.~\ref{fig:linMixMod}.  Here, a random vector $\xbf$ has
independent components with densities $p_{x_j}(x_j)$
and passes through a linear transform to yield
an output $\zbf=\Abf\xbf$.
The problem is to estimate $\xbf$ and $\zbf$
from measurements $\ybf$ generated
\textb{according to a conditional density $p_{\ybf|\zbf}(\ybf|\zbf)$
that is separable as a product of}
conditional densities $p_{y_i|z_i}(y_i|z_i)$.  Under this observation model,
the vectors $\xbf$ and $\zbf$ will have a posterior joint density given by
\beq \label{eq:pxzTrue}
    p_{\xbf,\zbf|\ybf}(\xbf,\zbf|\ybf) =
         [Z(\ybf)]^{-1}e^{-F(\xbf,\zbf)}\indic{\zbf=\Abf\xbf},
\eeq
where $F(\xbf,\zbf)$ is given by \eqref{eq:optsep} when the scalar functions
are set to the negative log prior denisty and likelihood:
\[
    f_{x_j}(x_j) = -\log p_{x_j}(x_j), \quad
    f_{z_i}(z_i) = -\log p_{y_i|z_i}(y_i|z_i).
\]
Note that in \eqref{eq:pxzTrue},
$F(\xbf,\zbf)$ is implicitly a function of $\ybf$,
$Z(\ybf)$ is a normalization constant, and
the point mass $\indic{\zbf=\Abf\xbf}$ imposes the linear
constraint that $\zbf=\Abf\xbf$.
The optimization \eqref{eq:xzOpt} in this case produces the
\emph{maximum a posteriori} (MAP) estimate of $\xbf$ and $\zbf$.
In statistics, the system in Fig.~\ref{fig:linMixMod}
is sometimes referred
to as a generalized linear model \iftoggle{conference}{\cite{McCulNel:89}}{\cite{NelWed:72,McCulNel:89}}
and is used in a range of applications
including regression, inverse problems, and filtering.  Bayesian
forms of compressed sensing can also be considered in this framework
by imposing a sparse prior for the components $x_j$~\cite{RanganFG:12-IT,Eldar:Book:12}.
In all these applications,
one may instead be interested in estimating the posterior marginals
$p(x_j|\ybf)$ and $p(z_i|\ybf)$.  We relate this objective to an optimization of the
form \eqref{eq:xzOpt}-\eqref{eq:optsep} in the sequel.

\begin{figure}

\center
\begin{tikzpicture}[scale=1]
    \node (x) {$\xbf \sim p_{\xbf}(\cdot)$};
    \node [block,node distance=2.5cm]  (A)   [right of=x] {$\Abf$ };
    \node [block,node distance=2.5cm]  (pyz) [right of=A] {$p_{\ybf|\zbf}(\cdot|\cdot)$}
        edge [<-] node[auto,swap] {$\zbf$} (A);
    \node [node distance=1.5cm] (y) [right of=pyz] {$\ybf$};

    \node [below of=x,font=\footnotesize] 
        {\parbox{2.0cm}{\centering Unknown input,\\ independent components} };
    \node [below of=A,font=\footnotesize] {Linear transform};
    \node [below of=pyz,text width=1.7cm,font=\footnotesize,xshift=-0.15cm]
        {\parbox{1.7cm}{\centering Componentwise\\ \hspace{0.1cm} output map} };

    \draw [->] (x) -- (A);
    \draw [->] (pyz) -- (y);
\end{tikzpicture}
\caption{System model:  The GAMP method considered here can
be used for approximate MAP and MMSE estimation of $\xbf$ from $\ybf$. \label{fig:linMixMod} }
\end{figure}
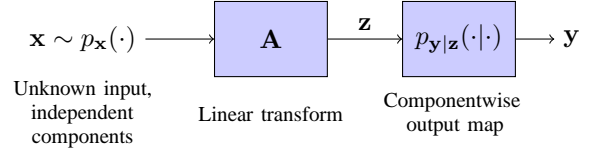

Most current numerical methods for solving the constrained
optimization problem \eqref{eq:xzOpt}
attempt to exploit the separable
structure of the objective function \eqref{eq:optsep}
either through generalizations of the iterative shrinkage and thresholding
(ISTA) algorithms
\iftoggle{conference}{\cite{ChamDLL:98,DaubechiesDM:04,WrightNF:09,BeckTeb:09,Nesterov:07,BioDFig:07}}{
\cite{ChamDLL:98,DaubechiesDM:04,VonUnser:04,WrightNF:09,BeckTeb:09,Nesterov:07,BioDFig:07}}
or alternating direction method of multipliers (ADMM) approach
\iftoggle{conference}{\cite{BoydPCPE:09,Esser:Diss:10,Chambolle:JMIV:11,He:JIS:12}}{
\cite{BoydPCPE:09,Eckstein:MP:92,Goldstein:JIS:09,Zhang:JSC:11,Combettes:MMS:05,
Tseng:91,Zhu:Tech:08,Esser:Diss:10,Chambolle:JMIV:11,He:JIS:12}}.
There are now many of these methods,
and we provide a brief review in Section~\ref{sec:optReview}.

However, in recent years, there has been considerable
interest in so-called approximate message passing (AMP) methods
based on Gaussian and quadratic approximations
of loopy belief propagation in graphical models
\cite{DonohoMM:09,DonohoMM:10-ITW1,DonohoMM:10-ITW2,BayatiM:11,Rangan:10-CISS,Rangan:11-ISIT}.
The main appealing feature of the AMP
algorithms is that for certain large random matrices $\Abf$,
the asymptotic
behavior of the algorithm can be rigorously and exactly predicted
with testable conditions for optimality, even for many non-convex
instances.
Moreover, in the case of these large, random matrices,
simulations appear to show very fast convergence of AMP methods
when compared
against state-of-the-art conventional optimization techniques.

Despite recent extensions to larger classes of random matrices \cite{CaireSTV:11,BayLelMon:12,cakmak2014samp},
the behavior of AMP methods under general $\Abf$ is not fully understood.
Indeed, for general $\Abf$, it is well-known that AMP methods may diverge \cite{RanSchFle:14-ISIT,Caltagirone:14-ISIT}.
While AMP has been successfully applied in a range of applications
\cite{FletcherRVB:11,chen2010improved,KamilovGR:12,vila2013hyperspectral,fletcher2014scalable},
the methods often require tuning to stabilize the algorithms.
Various general procedures to stabilize AMP have also been proposed~\cite{RanSchFle:14-ISIT,Vila:ICASSP:15,manoel2014swamp,Rangan:arxiv:15}.

To better understand these convergence issues,
the broad purpose of this paper is to show that certain forms of AMP algorithms
can be seen as variants of more conventional optimization methods.
This analysis will enable a precise characterization of the fixed points
of the AMP methods that applies to arbitrary $\Abf$, and a potential
framework to understand the convergence.

Our study focuses on a Generalized AMP (GAMP) method proposed
in \cite{Rangan:11-ISIT} and rigorously analyzed in
\cite{JavMon:12-arXiv}.
We consider this algorithm
since many other variants of AMP are
special cases of this general procedure.
The GAMP method has two common versions:  max-sum GAMP for
the MAP estimation of the vectors $\xbf$ and $\zbf$
for the problem in Fig.~\ref{fig:linMixMod}; and
sum-product GAMP for approximate inference of the
posterior marginals.

For both versions of GAMP, the algorithms produce estimates $\xbf$
and $\zbf$ along with certain ``quadratic'' terms.  Our first main result
 (Theorem~\ref{thm:gampMS}) shows that
the fixed points $(\xbfhat,\zbfhat)$ of max-sum GAMP
are critical points of the optimization \eqref{eq:xzOpt}.
In addition, the quadratic terms can be considered
as diagonal approximations of the inverse Hessian of the objective function.
For sum-product GAMP, we show (Theorem~\ref{thm:gampSP})
that the algorithm's fixed points are stationary points of a certain energy
function.

A conference version of this paper appeared in \cite{RanSRFC:13-ISIT}.
This paper includes all the proofs and more extensive discussion regarding relations
between GAMP and classic optimization and free energy minimization techniques.
In addition, since the publication of the conference version of this paper
in \cite{RanSRFC:13-ISIT}, several other works such as
\cite{Krzakala:14-ISITbethe,kabashima2014phase,RanSchFle:14-ISIT,Rangan:arxiv:15}
have built on the ideas and these are also discussed.

\section{Review of GAMP and Related Methods} \label{sec:optReview}

\subsection{Generalized Approximate Message Passing} \label{sec:gamp}

Graphical-model methods \cite{WainwrightJ:08}
are a natural approach to the optimization problem \eqref{eq:xzOpt}
given the separable structure of the objective function \eqref{eq:optsep}.
However, traditional graphical model techniques such as loopy belief propagation (loopy BP)
\textb{are computationally attractive only when}
the constraint matrix $\Abf$ is sparse.
Approximate message passing (AMP) refers to a class of Gaussian and quadratic approximations
of loopy BP that can be applied to dense $\Abf$.
AMP approximations of loopy BP originated in CDMA multiuser detection problems
\iftoggle{conference}{\cite{BoutrosC:02}}{\cite{BoutrosC:02,TanakaO:05,GuoW:06}} and have received considerable recent
attention in the context of compressed sensing \iftoggle{conference}{\cite{DonohoMM:09,DonohoMM:10-ITW1,DonohoMM:10-ITW2,BayatiM:11,Rangan:10-CISS, Rangan:11-ISIT}}{\cite{DonohoMM:09,DonohoMM:10-ITW1,DonohoMM:10-ITW2,BayatiM:11,Rangan:10-CISS, Rangan:11-ISIT,Montanari:12-bookChap}}.
The Gaussian approximations used in AMP are also closely related to
expectation propagation techniques \iftoggle{conference}{\cite{Minka:01}}{\cite{Minka:01,Seeger:08}}.

In this work, we study the so-called Generalized AMP (GAMP) algorithm
\cite{Rangan:11-ISIT} rigorously analyzed in \cite{JavMon:12-arXiv}.
The procedure, shown in Algorithm~\ref{algo:gamp}, produces
a sequence of estimates $(\xbf^t,\zbf^t)$ \textb{of $(\xbf,\zbf)$} along with the
\emph{quadratic terms} 
\textb{ $\taubf_x^t, \taubf^t_r \in \R_+^n$ and $\taubf^t_z, \taubf^t_p, \taubf^t_s \in \R_+^m$, where $t \in \Z_+$ represents the iteration number. 
Here and in the sequel, we use ``$.$'' to denote componentwise vector multiplication and ``$./$'' to denote componentwise vector division.}

\begin{algorithm}
\caption{Generalized Approximate Message Passing (GAMP)}
\begin{algorithmic}[1]  \label{algo:gamp}
\REQUIRE{ Matrix $\Abf\textb{\in\R^{m\times n}}$,
functions $f_x(\xbf),f_z(\zbf)\textb{\in\R}$,
and algorithm choice \texttt{MaxSum} or \texttt{SumProduct}. }

\STATE{ $t \gets 0$  }
\STATE{ Initialize $\xbf^t\textb{\in\R^n}$, $\taubf_x^t\textb{\in\R_+^n}$  }
\STATE{ $\sbf^{\tm1} \gets \textb{\mathbf{0}\in\R^m}$ }
\STATE{ $\Sbf \gets \Abf.\Abf$ (componentwise square)} \label{line:Sdef}
\REPEAT

    \STATE{ } \COMMENT{Output node update}
    \STATE{ $\taubf_p^t \gets \Sbf\taubf_x^t$ }
        \label{line:taup}
    \STATE{ $\pbf^t \gets \Abf\xbf^t -
        \sbf^{\tm1}.\taubf_p^t$ } \label{line:phat}
    \IF {\texttt{MaxSum}}
        \STATE{ $\zbf^t \gets \prox_{\taubf_p^t f_z}(\pbf^t)$}
            \label{line:zhatMS}
        \STATE{ $\taubf_z^t \gets \taubf_p^t . \prox'_{\taubf_p^t f_z}(\pbf^t)$}
            \label{line:tauzMS}
    \ELSIF {\texttt{SumProduct}}
        \STATE{ $\zbf^t \gets \Exp(\zbf|\pbf^t,\taubf_p^t)$}
            \label{line:zhatSP}
        \STATE{ $\taubf_z^t \gets \var(\zbf|\pbf^t,\taubf_p^t)$}     \label{line:tauzSP}
    \ENDIF
    \STATE{ $\sbf^t \gets (\zbf^t-\pbf^t)./\taubf_p^t$ }
            \label{line:shat}
    \STATE{ $\taubf_s^t \gets (\mathbf{1}- \taubf_z^t./\taubf_p^t)./\taubf_p^t $ }
           \label{line:taus}

    \STATE{ }
    \STATE{ } \COMMENT{Input node update}
    \STATE{ $\taubf_r^t \gets \mathbf{1}./(\Sbf^T\taubf_s^t)$ }
            \label{line:taur}
    \STATE{ $\rbf^t \gets \xbf^t + \taubf_r^t .\Abf^T\sbf^t$ } \label{line:rhat}
    \IF {\texttt{MaxSum}}
        \STATE{ $\xbf^{\tp1} \gets \prox_{\taubf_r^t f_x}(\rbf^t)$}
	    \label{line:xhatMS}
        \STATE{ $\taubf_x^{\tp1} \gets \taubf_r^t . \prox'_{\taubf_r^t f_x}(\rbf^t)$}
            \label{line:tauxMS}
    \ELSIF {\texttt{SumProduct}}
        \STATE{ $\xbf^{\tp1} \gets \Exp(\xbf|\rbf^t,\taubf_r^t)$}
            \label{line:xhatSP}
        \STATE{ $\taubf_x^{\tp1} \gets \var(\xbf|\rbf^t,\taubf_r^t)$}
            \label{line:tauxSP}
    \ENDIF

\UNTIL{Terminated}

\end{algorithmic}
\end{algorithm}

We focus on two variants of the GAMP algorithm: \emph{max-sum GAMP} and
\emph{sum-product GAMP}.

\paragraph*{Max-sum GAMP}
In the max-sum version of the algorithm,
the outputs $(\xbf^t,\zbf^t)$ represent
estimates of the solution to the optimization problem \eqref{eq:xzOpt},
or equivalently the MAP estimates for the posterior \eqref{eq:pxzTrue}.
Since the objective function has the separable form \eqref{eq:optsep},
each iteration of the algorithm involves four componentwise update steps:
the proximal updates shown in lines \ref{line:zhatMS} and \ref{line:xhatMS},
where
\beq \label{eq:proxDef}
    \prox_f(v) := \argmin_{u \in \R} f(u) + \frac{1}{2}(u-v)^2,
\eeq
and lines \ref{line:tauzMS} and \ref{line:tauxMS},
involving the derivative of the proximal operator from \eqref{eq:proxDef}.

In particular, lines \ref{line:zhatMS} and \ref{line:tauzMS} are to be interpreted as
\begin{align}
    z_i^t &= \prox_{\tau_{p_i}^t f_{z_i}}(p_i^t), ~~~i=1,\dots,m,\\
    \tau_{z_i}^t &= \tau_{p_i}^t\prox'_{\tau_{p_i}^t f_{z_i}}(p_i^t), ~~~i=1,\dots,m,\\
    &=\tau_{p_i}^t\Big(1+\tau_{p_i}^t \frac{\partial^2 f_{z_i}(z_i^t)}{\partial z_i^2}\Big)^{-1}, ~~~i=1,\dots,m, \label{eq:proxderiv}
\end{align}
with similar interpretations for lines \ref{line:xhatMS} and \ref{line:tauxMS}.
Thus, max-sum GAMP reduces the vector-valued optimization \eqref{eq:xzOpt} to a sequence of scalar optimizations.

\textb{When discussing max-sum GAMP, we will assume that both $f_x$ and $f_z$ are twice differentiable and convex, so that the outputs of the proximal operator and its derivative exist and are unique.
We make these assumptions for the sake of clarity, but note that---in practice---GAMP is often used with non-differentiable functions.
A common example is when $f_x(\xbf)=\lambda\|\xbf\|_1$ for $\lambda>0$, in which case 
\begin{align}
\prox_{\tau_{r_j}^t f_{x_j}}(r_j^t) 
&= \sgn(r_j^t)\max\{|r_j^t|-\lambda\tau_{r_j}^t,0\}
\label{eq:soft_thresh}
\end{align}
and 
\begin{align}
\prox'_{\tau_{r_j}^t f_{x_j}}(r_j^t) 
= \begin{cases}
1, & |r_j^t| > \lambda\tau_{r_j}^t; \\
0, & |r_j^t| < \lambda\tau_{r_j}^t .
\end{cases} 
\end{align}
Although $\prox'_{\tau_{r_j}^t f_{x_j}}(r_j^t)$ is undefined when $r_j^t=\lambda\tau_{r_j}^t$, its value can be set to either $0$ or $1$ with minimal effect, because the event $r_j^t=\lambda\tau_{r_j}^t$ almost never occurs (due, e.g., to the presence of noise in $r_j^t$).
The rigorous GAMP analysis \cite{JavMon:12-arXiv} assumes only that the prox functions in lines \ref{line:zhatMS} and \ref{line:xhatMS} are Lipschitz continuous (and hence differentiable almost everywhere).
}

\paragraph*{Sum-product GAMP}
The purpose of the
sum-product GAMP algorithm is to provide estimates of the posterior marginals
\beq \label{eq:pxzMarg}
    p(x_j|\ybf), \quad p(z_i|\ybf),
\eeq
from the joint density \eqref{eq:pxzTrue}.  Exact computation of these
marginal densities is, in general, computationally intractable.
Sum-product GAMP instead provides estimates of these densities.
Specifically, at each iteration $t$, it forms the estimated densities,
called \emph{beliefs}, given by:
\beq \label{eq:bxzSep}
    b^t_{x_j}(x_j) = p(x_j|r_j^t,\tau_{r_j}^t), \quad
    b^t_{z_i}(z_i) = p(z_i|p_i^t,\tau_{p_i}^t),
\eeq
where we use the notation
\begin{subequations} \label{eq:pxzrpSep}
\beqa
    p(x_j|r_j,\tau_{r_j}) & \propto & \exp\left[ -f_{x_j}(x_j) -
        \textstyle\frac{1}{2\tau_{r_j}}(x_j-r_j)^2 \right], \\
    p(z_i|p_i,\tau_{p_i}) & \propto & \exp\left[ -f_{z_i}(z_i) -
        \textstyle\frac{1}{2\tau_{p_i}}(z_i-p_i)^2 \right].
\eeqa
\end{subequations}
As we will discuss in Section~\ref{sec:sumProdFP},
these belief  estimates can be ``derived" as estimates of the minima
of a certain large system limit of the Bethe Free Energy.

Now, the products of the densities in \eqref{eq:pxzrpSep}
are given by
\begin{subequations} \label{eq:pxzrp}
\beqa
    p(\xbf|\rbf,\taubf_r) &=& \prod_{j=1}^n
    p(x_j|r_j,\tau_{r_j}) \nonumber \\
    &\propto& \exp\left[ -f_{x}(\xbf) -
        \frac{1}{2}\|\xbf-\rbf\|_{\taubf_r}^2 \right], \\
    p(\zbf|\pbf,\taubf_p) &=&
    \prod_{i=1}^m p(z_i|z_i,\tau_{p_i}) \nonumber \\
    &\propto&
    \exp\left[ -f_{z}(\zbf) -
        \frac{1}{2}\|\zbf-\pbf\|_{\taubf_p}^2 \right],
\eeqa
\end{subequations}
where, for any vectors $\vbf\in\R^r$ and $\taubf \in \R^r$ with $\taubf>0$,
we use the notation
\[
    \|\vbf\|^2_\taubf := \sum_{i=1}^r \frac{|v_i|^2}{\tau_i}.
\]
In the sum-product version of GAMP,
the expectations and variances in lines
\ref{line:zhatSP}, \ref{line:tauzSP},
\ref{line:xhatSP} and \ref{line:tauxSP}
of Algorithm~\ref{algo:gamp}
are to be
taken with respect to the probability density functions in
\eqref{eq:pxzrp}.
Thus, $\xbf^t$ and $\taubf_x^t$ are the estimates of the posterior
means and variances of the components of $\xbf$ and
$\zbf^t$ and $\taubf_z^t$ are the
estimates of the posterior means and variances of the components of $\zbf$.

Since the densities \eqref{eq:pxzrp} are separable, the expectations
and variances can be computed via scalar integrals.  Thus, the sum-product
GAMP algorithm reduces the vector-valued to marginalization problem
to a sequence of scalar estimation problems.

\subsection{Iterative Shrinkage and Thresholding Algorithm} \label{sec:ista}

The goal in the paper is to relate the GAMP method to more conventional
optimization techniques.  One of the more common of such approaches is a generalization
of the Iterative Shrinkage and Thresholding Algorithm (ISTA) shown in
Algorithm~\ref{algo:ISTA} \iftoggle{conference}{\cite{ChamDLL:98,DaubechiesDM:04,WrightNF:09,BeckTeb:09}}{
\cite{ChamDLL:98,DaubechiesDM:04,VonUnser:04,WrightNF:09,BeckTeb:09}},
\textb{where $\gradient f$ denotes the gradient of $f$.}

\begin{algorithm}
\caption{Iterative Shrinkage and Thresholding Algorithm (ISTA)}
\begin{algorithmic}[1]  \label{algo:ISTA}
\REQUIRE{ Matrix $\Abf$, scalar $c\geq 0$, functions $f_x(\cdot)$, $f_z(\cdot)$. }

\STATE{ $t \gets 0$  }
\STATE{ Initialize $\xbf^t$.  }
\REPEAT
    \STATE{ $\zbf^t \gets \Abf\xbf^t$ } \label{line:zISTA}
    \STATE{ $\qbf^t \gets \gradient f_z(\zbf^t)$ } \label{line:qISTA}
    \STATE{ $\xbf^{\tp1} \gets \argmin_{\xbf} f_x(\xbf) + (\qbf^t)^T\Abf\xbf +
        (c/2)\|\xbf-\xbf^t\|^2$ } \label{line:xminISTA}
\UNTIL{Terminated}
\end{algorithmic}
\end{algorithm}

The algorithm is built on the idea that, at each iteration $t$, the second cost term in the minimization $\argmin_{\xbf}f_x(\xbf)+f_z(\Abf\xbf)$ specified by \eqref{eq:xzOpt} is replaced by a quadratic majorizing cost $g_z(\xbf)\geq f_z(\Abf\xbf)$ that coincides at the point $\xbf=\xbf^t$ (i.e., $g_z(\xbf^t)=f_z(\Abf\xbf^t)$).
\textb{The function $g_z(\xbf)$ defined implicitly in line \ref{line:xminISTA} achieves this}
majorization via appropriate choice of $c>0$.
This approach is motivated by the fact that, if $f_x(\xbf)$ and $f_z(\zbf)$ are both
separable, as in \eqref{eq:optsep}, then both the gradient in line \ref{line:qISTA}
and minimization in line \ref{line:xminISTA} can be performed componentwise.
Moreover, when $f_x(\xbf) = \lambda\|\xbf\|_1$, as in the LASSO problem \cite{Tibshirani:96}, the minimization in line \ref{line:xminISTA} can be computed directly via the
shrinkage and thresholding operation \textb{\eqref{eq:soft_thresh}}---hence the name of the algorithm.
The convergence of the ISTA method tends to be slow, but a number of enhanced methods
have been successful and widely-used
\cite{WrightNF:09,BeckTeb:09,Nesterov:07,BioDFig:07}.

\subsection{Alternating Direction Method of Multipliers}
\label{sec:alMeth}

A second common class of methods is built around the Alternating Direction Method of Multipliers (ADMM) \cite{BoydPCPE:09} approach
shown in Algorithm~\ref{algo:ADMM}.
The Lagrangian for the optimization problem \eqref{eq:xzOpt} is given by
\beq \label{eq:Lag}
    L(\xbf,\zbf,\sbf) := F(\xbf,\zbf) + \sbf^T(\zbf-\Abf\xbf),
\eeq
where $\sbf$ are the dual parameters.
ADMM attempts to produce a sequence of estimates $(\xbf^t,\zbf^t,\sbf^t)$
that converge to a saddle point of the Lagrangian \eqref{eq:Lag}.
The parameters of the algorithm are a step-size $\alpha>0$
and the penalty terms $Q_z(\cdot)$ and $Q_x(\cdot)$, which
classical ADMM would choose as
\begin{subequations}\label{eq:QADMM}
\beqa
     Q_x(\xbf,\xbf^t,\zbf^t,\alpha) &=& \frac{\alpha}{2}\|\zbf^t-\Abf\xbf\|^2 \label{eq:QxADMM} \\
     Q_z(\zbf,\zbf^t,\xbf^{\tp1},\alpha) &=& \frac{\alpha}{2}\|\zbf-\Abf\xbf^{\tp1}\|^2. \label{eq:QzADMM}
\eeqa
\end{subequations}
\vspace{-3mm}

\begin{algorithm}
\caption{Alternating Direction Method of Multipliers (ADMM)}
\begin{algorithmic}[1]  \label{algo:ADMM}
\REQUIRE{ $\Abf$, $\alpha$, functions $f_x(\cdot)$, $f_z(\cdot)$, $Q_x(\cdot)$, $Q_z(\cdot)$ }

\STATE{ $t \gets 0$  }
\STATE{ Initialize $\xbf^t$, $\zbf^t$, $\sbf^t$  }
\REPEAT
    \STATE{ $\xbf^{\tp1} \gets \argmin_{\xbf} L(\xbf,\zbf^t,\sbf^t)
        + Q_x(\xbf,\xbf^t,\zbf^t,\alpha)$ } \label{line:xminAL}
    \STATE{ $\zbf^{\tp1} \gets \argmin_{\zbf} L(\xbf^{\tp1},\zbf,\sbf^t)
        + Q_z(\zbf,\zbf^t,\xbf^{\tp1},\alpha)$ } \label{line:zminAL}
    \STATE{ $\sbf^{\tp1} \gets \sbf^t + \alpha(\zbf^{\tp1}-\Abf\xbf^{\tp1})$ }
        \label{line:sAL}
\UNTIL{Terminated}

\end{algorithmic}
\end{algorithm}

When the objective function admits a separable form \eqref{eq:optsep} and one uses the auxiliary function $Q_z(\cdot)$ in \eqref{eq:QzADMM}, the $\zbf$-minimization in line \ref{line:zminAL} separates into $m$ scalar optimizations.
However, due to the quadratic term $\|\Abf\xbf\|^2$ in \eqref{eq:QxADMM}, the $\xbf$-minimization in line \ref{line:xminAL} does not separate for general $\Abf$.
To circumvent this problem, one might consider a separable inexact $\xbf$-minimization, since many inexact variants of ADMM are known to converge \iftoggle{conference}{\cite{BoydPCPE:09}}{\cite{Eckstein:MP:92}}.
For example, $Q_x(\cdot)$ might be chosen to yield separability while majorizing the original cost in line \ref{line:xminAL}, as was done for ISTA's line \ref{line:xminISTA}%
\iftoggle{conference}{.}{, i.e.,
\begin{align}
\lefteqn{Q_x(\xbf,\xbf^t,\zbf^t,\alpha)}\\
&= \frac{\alpha}{2}\|\zbf^t-\Abf\xbf\|^2
+\frac{1}{2}(\xbf-\xbf^t)^T(c\Ibf-\alpha\Abf^T\Abf)(\xbf-\xbf^t) \nonumber
\end{align}
with \textb{$c\geq \alpha\|\Abf\|^2$}, after which ADMM's line \ref{line:xminAL} would become
\begin{align}\label{eq:xminIAL}
\argmin_{\xbf} f_x(\xbf) + \frac{c}{2}\Big\|\xbf-\xbf^t+\frac{\alpha}{c}\Abf^T\Big(\Abf\xbf^t-\zbf^t-\frac{1}{\alpha}\sbf^t\Big)\Big\|^2 .
\end{align}
\textb{This approach is known as ``linearized ADMM'' \cite{ParikhB:13}, or as ``split inexact Uzawa'' \cite{Zhang:JSC:11} in the optimization literature}, and it has close connections to other well-known techniques like
Douglas--Rachford splitting \cite{Eckstein:MP:92},
split Bregman \cite{Goldstein:JIS:09},
proximal forward-backward splitting \cite{Combettes:MMS:05},
and various primal-dual algorithms \cite{Tseng:91,Zhu:Tech:08,Esser:Diss:10,Chambolle:JMIV:11,He:JIS:12}.}
Many other choices of penalty $Q_x(\cdot)$ have also been considered in the literature (see, e.g., the overview in \cite{Esser:Diss:10}).

Other variants of ADMM are also possible \cite{BoydPCPE:09}.
For example, the step-size $\alpha$ might vary with the iteration $t$,
or the penalty terms might have the form $(\zbf-\Abf\xbf)^T\Pbf(\zbf-\Abf\xbf)$ for positive semidefinite $\Pbf$.
As we will see, these generalizations provide a connection to GAMP.

\section{Fixed-Points of Max-Sum GAMP} \label{sec:maxSumFP}

Our first result connects the max-sum GAMP algorithm
to inexact ADMM.
\iftoggle{conference}{}{  Given points $(\xbf,\zbf)$, define the matrices
\begin{subequations} \label{eq:Qxz}
\beqa
\label{eq:Qxz_a}
    \Qbf_x &:=& \Bigl(\Diag(\dbf_x)
        + \Abf^T\Diag(\dbf_z)\Abf \Bigr)^{-1} \\
    \Qbf_z &:=& \Bigl(\Diag(\dbf_z)^{-1}
        + \Abf\Diag(\dbf_x)^{-1}\Abf^T \Bigr)^{-1}
\eeqa
\end{subequations}
where \textb{$\Diag(\dbf)$ denotes the diagonal matrix with diagonal entries equal to those in the vector $\dbf$}, and where $\dbf_x$ and $\dbf_z$ contain the componentwise second derivatives, \textb{i.e., the diagonals of the Hessian matrices}
\beq \label{eq:dxz}
    \dbf_x := \diag\left[ \Hessian f_x(\xbf) \right], \quad
    \dbf_z := \diag\left[ \Hessian f_z(\zbf) \right].
\eeq
Note that when $f_x$ and $f_z$ are strictly convex,
the elements in $\dbf_x$ and $\dbf_z$ are positive.
Observe that
the matrix $\Qbf_x$ in \eqref{eq:Qxz_a} is the inverse
Hessian of the objective function $F(\xbf,\zbf)$ constrained
to $\zbf=\Abf\xbf$. That is,
\[
    \Qbf_x = \left[ \Hessian_{\xbf} F(\xbf,\Abf\xbf) \right]^{-1}.
\]
}

\begin{theorem} \label{thm:gampMS}
The outputs of the max-sum GAMP version of Algorithm~\ref{algo:gamp}
satisfy the recursions
\begin{subequations}  \label{eq:lagMS}
\beqa
    \xbf^{\tp1} &=& \argmin_{\xbf} \Big[
    L(\xbf,\zbf^t,\sbf^t)  +
     \frac{1}{2}\|\xbf-\xbf^t\|^2_{\taubf_r^t} \Big]
     \label{eq:xiterMS} \\
    \zbf^{\tp1} &=& \argmin_{\zbf} \Big[ L(\xbf^{\tp1},\zbf,\sbf^t)
    + \frac{1}{2} \|\zbf-\Abf\xbf^{\tp1}\|^2_{\taubf_p^{\tp1}}  \Big]
        \qquad \label{eq:ziterMS}\\
    \sbf^{\tp1} &=& \sbf^t + (\zbf^{\tp1}-\Abf\xbf^{\tp1})./\taubf_p^{\tp1}
      \label{eq:siterMS}
\eeqa
\end{subequations}
where $L(\xbf,\zbf,\sbf)$ is the Lagrangian defined in \eqref{eq:Lag}.

Now suppose that $(\xbfhat,\zbfhat,\sbf,\taubf_x,\taubf_s)$
is a fixed point of the algorithm (where the ``hats" on $\xbfhat$ and $\zbfhat$
are used to distinguish them from free variables).  Then, this fixed point is a critical point of the constrained optimization \eqref{eq:xzOpt} in that $\zbfhat = \Abf\xbfhat$
and
\beq \label{eq:LagFPMS}
    \gradient_{\xbf}
        L(\xbfhat,\zbfhat,\sbf) = \zero, \quad
    \gradient_{\zbf}
        L(\xbfhat,\zbfhat,\sbf) = \zero.
\eeq

Moreover, the quadratic terms $\taubf_x,\taubf_s$ are the \emph{approximate
diagonals} (as defined in Appendix~\ref{sec:approxDiag}) of
$\Qbf_x$ and
$\Qbf_z$ in \eqref{eq:Qxz}
at $(\xbf,\zbf)=(\xbfhat,\zbfhat)$.
\end{theorem}
\begin{IEEEproof}
See Appendix \ref{sec:gampMSPf}.
\end{IEEEproof}

The first part of the theorem, equations \eqref{eq:lagMS}, shows that max-sum GAMP can be interpreted as the ADMM Algorithm~\ref{algo:ADMM} with adaptive vector-valued step-sizes $\taubf_r^t$ and $\taubf_p^t$ and a particular choice of penalty $Q_x(\cdot)$.
\iftoggle{conference}{Although the convergence of a simpler variant of this algorithm, with non-adaptive step-sizes and convex $f_x(\cdot)$ and $f_z(\cdot)$, was studied in \cite{Esser:Diss:10,Chambolle:JMIV:11,He:JIS:12}, the results there do not directly apply to the more general GAMP algorithm.}{
To more precisely connect GAMP and existing algorithms, it helps to express GAMP's $\xbf$-update \eqref{eq:xiterMS} as the $\theta\!=\!0$ case of
\begin{align}
\argmin_{\xbf} f_x(\xbf) + \frac{1}{2}\big\|\xbf-\xbf^t+\taubf_r^t .\Abf^T\big( \theta(\sbf^{t-1}-\sbf^t)-\sbf^t\big)\big\|_{\taubf_r^t}^2 ,
\end{align}
and recognize that the ISTA-inspired inexact ADMM $\xbf$-update \eqref{eq:xminIAL} coincides with the $\theta\!=\!1$ case under step-sizes $\alpha=1/\taubf_p^t$ and $c=1/\taubf_r^t$.
The convergence of this algorithm for particular $\theta\in[0,1]$ was studied in \cite{Esser:Diss:10,Chambolle:JMIV:11,He:JIS:12} under convex functions $f_x(\cdot)$ and $f_z(\cdot)$ and non-adaptive step-sizes.
Unfortunately, these convergence results do not directly apply to the adaptive vector-valued step-sizes of GAMP.

The second part of the theorem, \textb{equation \eqref{eq:LagFPMS}}, shows that if the algorithm converges then its fixed points will be critical points of the constrained optimization \eqref{eq:xzOpt}.
\textb{This part of the theorem can be considered as a generalization of Proposition~7.1 in \cite{DonohoJM:11arXiv}, which considers quadratic $f_z$, and of Proposition~5.1 in \cite{Montanari:12-bookChap}, which considers quadratic $f_z$ and $f_x(\xbf)=\|\xbf\|_1$.} }

\iftoggle{conference}{The full paper \cite{RanSRFC:13-arxiv} also shows that the quadratic
term $\taubf_x$ can be interpreted as an approximate diagonal to the inverse Hessian.}{
\textb{The third part of Theorem~\ref{thm:gampMS} then shows that the quadratic term $\taubf_x$ can be interpreted as an ``approximate diagonal'' of the inverse Hessian under the large random matrix model described in Appendix~\ref{sec:approxDiag}.}

Finally, it is useful to compare the fixed-points of GAMP with those of standard BP\@.
A classic result of \cite{WeissFree:01} shows that
any fixed point for standard max-sum loopy BP is locally optimal in the sense that
one cannot improve the objective function by perturbing the solution on any set of components
whose variables belong to a subgraph that contains at most one cycle.  In particular, if the
overall graph is acyclic, any fixed-point of standard max-sum loopy BP
is globally optimal.  Also, for any graph, the objective
function cannot be reduced by changing any individual component.  The local optimality
for GAMP provided by Theorem~\ref{thm:gampMS}
is weaker \textb{than that for max-sum loopy BP} in that \textb{GAMP's} fixed-points only satisfy first-order conditions for
saddle points of the Lagrangian.  This implies that, even an individual component may only
be locally optimal.
}

\section{Fixed-Points of Sum-Product GAMP} \label{sec:sumProdFP}

\subsection{Bethe Free Energy} \label{sec:BFE}

A classic result in graphical models is that the fixed points of loopy
BP can be interpreted as critical points in the constrained
minimization of a energy function known as the Bethe Free energy (BFE)
\cite{YedidiaFW:03,Heskes:03}.
In this section, we will show that sum-product
GAMP has a similar energy function interpretation.

Specifically, consider a set of scalar densities
\beq \label{eq:bqxz}
    b_{x_j}(x_j), \quad b_{z_i}(z_i), \quad q_{z_i}(z_i),
\eeq
where the densities $q_{z_i}(z_i)$ are Gaussian.
Given any such set, define the product densities
\begin{subequations} \label{eq:bqxzProd}
\beqa
    & & b_x(\xbf) = \prod_{j=1}^n b_{x_j}(x_j), \quad
    b_z(\zbf) = \prod_{i=1}^m b_{z_i}(z_i)  \label{eq:bxzProd}\\
    & & q_z(\zbf) = \prod_{i=1}^m q_{z_i}(z_i), \label{eq:qzProd}
\eeqa
\end{subequations}
and the energy function
\beqa
    J_{\rm SP}(b_x,b_z,q_z) & := & D(b_x\|e^{-f_x}) + D(b_z\|e^{-f_z}) \nonumber \\
    && + D(b_z\|q_z) + H(b_z),  \hspace{3cm} \label{eq:JBFEq}
\eeqa
where $H(b_z)$ is the differential entropy.
With these definitions,
consider the constrained minimization
\beq \label{eq:JBFEoptq}
\begin{aligned}
    & \min_{b_x,b_z,q_z} & & J_{\rm SP}(b_x,b_z,q_z)  \\
    & \text{s.t.} & &   \Exp(\zbf|b_z) = \Exp(\zbf|q_z) = \Abf\Exp(\xbf|b_x) \\
    & & & \taubf_p = \Sbf \var(\xbf|b_x), \quad \Sbf = \Abf.\Abf  \\
    & & & q_z(\zbf) \sim {\mathcal N}\left(\zbf|\mubf_p,\Diag(\taubf_p)\right),
\end{aligned}
\eeq
Here and below, \textb{we use $\Exp(\xbf|b_x)$ to denote the expected value of $\xbf\sim b_x$, and similar for $\Exp(\zbf|b_z)$.
Also, we use $\var(\xbf|b_x)$ to denote the vector whose $j$th component is the variance of $x_j\sim b_{x_j}$, and similar for $\var(\zbf|b_z)$.
We stress that $\var(\xbf|b_x)$ is a vector, not a covariance matrix.}
Note also that the last constraint in
\eqref{eq:JBFEoptq} simply states that $q_z$ must be Gaussian with independent
components.

Note that since
\[
    D(b_z\|q_z) + H(b_z) = -\Exp\left[ \log q_z(\zbf) \mid b_z \right],
\]
the objective function \eqref{eq:JBFEq} is separately convex in $(b_x,b_z)$ and
$q_z$.  However, it is not, in general, jointly convex in all three densities.  Also, 
the final two constraints in the optimization \eqref{eq:JBFEoptq}, on the variances
and Gaussianity of $q_z$, are also not convex.  

Our main result, Theorem~\ref{thm:gampSP} below, shows that
sum-product GAMP can be interpreted as a method to approximately
minimize this non-convex energy function.
This result was first stated in the conference version of this
paper~\cite{RanSRFC:13-ISIT}.  Since the publication of that paper,
it was stated in \cite{Krzakala:14-ISITbethe} that, in the case
of additive white Gaussian noise (AWGN) output channels, the constrained
optimization
\eqref{eq:JBFEoptq} can be interpreted as an approximation of
the Bethe Free energy optimization that is valid when (a) the matrix $\Abf$ has
i.i.d.\ zero mean entries and $m,n \arr \infty$, and (b) the standard
marginalization constraints in the BFE optimization are replaced by
matching constraints on the first and second moments.
A subsequent work
\cite{kabashima2014phase} derived a similar approximate BFE optimization
for arbitrary output channels and matrix uncertainties.
We will not discuss the BFE interpretation in this work; the reader
is referred to \cite{Krzakala:14-ISITbethe,kabashima2014phase}.
However, in recognition of the relation to the Bethe free energy minimization,
we will call the energy function \eqref{eq:JBFEq}
the large system limit Bethe Free energy (LSL-BFE) and call the constrained
minimization
\eqref{eq:JBFEoptq} the LSL-BFE optimization.

\subsection{GAMP Optimization}

To relate the LSL-BFE optimization
\eqref{eq:JBFEoptq} to sum-product GAMP,  we first rewrite the optimization
to remove the minima over $q_z$.
Given a density $b_z(\zbf)$, define the function
\beqa
   \lefteqn{ H_{\rm gauss}(b_z,\taubf_p) :=  D(b_z\|q_z) + H(b_z), } \nonumber \\
    &  q_z(\zbf) = {\mathcal N}(\zbf|\mubf_p,\Diag(\taubf_p)), \quad
    \mubf_p = \Exp(\zbf|b_z).  \label{eq:Hgauss1}
\eeqa
This function is simply the last two terms of $J_{\rm SP}(b_x,b_z,q_z)$
in \eqref{eq:JBFEq}
with $q_z(\zbf)$ being the Gaussian density with mean $\mubf_p = \Exp(\zbf|b_z)$
and variance $\var(\zbf|q_z) = \taubf_p$.
It can be calculated that
\beq \label{eq:Hgauss}
     H_{\rm gauss}(b_z,\taubf_p)
     = \frac{1}{2}\sum_{i=1}^m \left[ \frac{\var(z_i|b_{z_i})}{\tau_{p_i}} + \log(2\pi\tau_{p_i}) \right].
\eeq
Note that from \eqref{eq:Hgauss1}, $H_{\rm gauss}(b_z,\taubf_p) \geq H(b_z)$
for all $\taubf_p$ with equality when $b_z$ is itself Gaussian with variance
$\var(\zbf|b_z) = \taubf_p$.  Hence, we will call $H_{\rm gauss}(b_z,\taubf_p)$
the \emph{Gaussian entropy upper bound} function.
Using this upper bound function, we can replace
the minimization over Gaussian $q_z$ in
\eqref{eq:JBFEoptq} with an optimization over the vector of variances
$\taubf_p$.  This results in the equivalent optimization
\beq \label{eq:JBFEopt}
\begin{aligned}
    & \min_{b_x,b_z,\taubf_p} & & J_{\rm SP}(b_x,b_z,\taubf_p)  \\
    & \text{s.t.} & &   \Exp(\zbf|b_z) = \Abf\Exp(\xbf|b_x) \\
    & & & \taubf_p = \Sbf \var(\xbf|b_x), \quad \Sbf = \Abf.\Abf
\end{aligned}
\eeq
where the objective function is
\beqa
    J_{\rm SP}(b_x,b_z,\taubf_p) &:=&
        D(b_x\|e^{-f_x})  + D(b_z\|e^{-f_z})
	\nonumber\\&&\mbox{}
      + H_{\rm gauss}(b_z,\taubf_p). \label{eq:JSP}
\eeqa
With some abuse of notation, we have used $J_{\rm SP}(\cdot)$
to denote both the LSL-BFE function in terms of $q_z$ as in~\eqref{eq:JBFEq}
and the function in terms of the variance vector $\taubf_p$ as given in
\eqref{eq:JSP}.

Corresponding to \eqref{eq:JBFEopt}, define the Lagrangian
\beqa
    L_{\rm SP}(b_x,b_z,\taubf_p,\sbf)
    & = & J_{\rm SP}(b_x,b_z,\taubf_p) \nonumber \\
    & & + \sbf^T(\Exp(\zbf|b_z) - \Abf\Exp(\xbf|b_x)), \label{eq:LSP}
\eeqa
where $\sbf$ represents a vector of dual parameters.
Note that this Lagrangian does \emph{not} include the constraint
$\taubf_p = \Sbf\var(\xbf|b_x)$; we will handle that separately.
We can now state the main result.

\begin{theorem} \label{thm:gampSP}
Consider the outputs of the sum-product GAMP version of Algorithm~\ref{algo:gamp}, and define the densities
\beq \label{eq:bxzSP}
    b_x^{\tp1}(\xbf) = p(\xbf|\rbf^t,\taubf_r^t), \quad
    b_z^t(\zbf) = p(\zbf|\pbf^t,\taubf_p^t),
\eeq
where $p(\xbf|\rbf,\taubf_r)$ and $p(\zbf|\pbf,\taubf_p)$ are given in \eqref{eq:pxzrp}.
Then, the GAMP algorithm input node update satisfies
\beqa
   b_x^{\tp1} &=& \argmin_{b_x}
   \Big[ L_{\rm SP}(b_x,b_z^t,\taubf_p^t,\sbf^t)
    + \frac{1}{2}(\taubf_s^t)^T\Sbf\var(\xbf|b_x)
    \nonumber \\ & & \mbox{}
    + \frac{1}{2}\left\|\Exp(\xbf|b_x)-\Exp(\xbf|b_x^t)
        \right\|^2_{\taubf_r^t} \Bigr].
    \label{eq:lagInSP}
\eeqa
where $L_{\rm SP}(\xbf,\zbf,\sbf)$ is the Lagrangian in \eqref{eq:LSP}.
Similarly, the steps in the output node update for the GAMP algorithm
are equivalent to:
\begin{subequations}  \label{eq:lagOutSP}
\beqa
    \taubf_p^t &=& \Sbf\var(\xbf|b_x^t), \label{eq:taupSP} \\
   b_z^t &=& \argmin_{b_z}
   \Big[ L_{\rm SP}(b_x^t,b_z,\taubf_p^t,\sbf^{\tm1}) \nonumber \\
   & & \mbox{}
	+ \frac{1}{2} \left\|\Exp(\zbf|b_z)-\Abf\Exp(\xbf|b_x^t)\right\|^2_{\taubf_p^{t}}
        \Bigr], \label{eq:bzSP} \\
    \sbf^{t} &=& \sbf^{\tm1} + \frac{1}{\taubf_p^{t}}
    \left[ \Exp(\zbf|b_z^{t})-\Abf\Exp(\xbf|b_x^{t}) \right],
    \label{eq:siterSP}\\
    \taubf_s^{t} &=& 2 \gradient_{\taubf_p}
    L_{\rm SP}(b_x^t,b_z^t,\taubf_p^{t},\sbf^{t}).
    \label{eq:tausiterSP}
\eeqa
\end{subequations}
Moreover, any fixed point of the sum-product GAMP algorithm
is a critical point of the constrained optimization \eqref{eq:JBFEopt}.
\end{theorem}
\begin{proof} See Appendix \ref{sec:gampSPPf}.
\end{proof}

Theorem~\ref{thm:gampSP} exposes connections between sum-product GAMP and both the
ISTA and ADMM methods described earlier.
The minimizations over $b_x$ and $b_z$ and the update of the dual parameters
$\sbf^t$ in \eqref{eq:lagInSP}, \eqref{eq:bzSP} and \eqref{eq:siterSP}
follow the format of the ADMM minimizations in Algorithm~\ref{algo:ADMM}
for certain choices of the auxiliary functions.  On the other hand, the
role of $\taubf_s^t$ in \eqref{eq:lagInSP} and \eqref{eq:tausiterSP}
follows the gradient-based method of the
generalized ISTA method in Algorithm \ref{algo:ISTA} for the constraint
$\taubf_s=\Sbf\var(\xbf|b_x)$.  So, the sum-product GAMP
algorithm can be seen as a hybrid of the ISTA and ADMM methods for the
optimization problem \eqref{eq:JBFEopt}.

Unfortunately, this hybrid ISTA-ADMM method is non-standard and we are not aware of existing
convergence theory. However, Theorem~\ref{thm:gampSP} at least shows that,
if the sum-product GAMP algorithm converges, then
its fixed points correspond to critical points of the optimization problem \eqref{eq:JBFEopt}.

\section*{Conclusions}
Although AMP methods admit precise analyses in the context of large
i.i.d.\ transform matrices $\Abf$, their behavior for general matrices is less well-understood.
This limitation is unfortunate since many transforms arising in practical problems such as imaging and regression
are not well-modeled as realizations of
large i.i.d.\ matrices.
To help overcome these limitations, this paper draws connections between AMP and certain variants of standard optimization methods that employ adaptive vector-valued step-sizes.
These connections enable a precise characterization of the fixed-points
of both max-sum and sum-product GAMP for the case of arbitrary transform matrices $\Abf$.

However, much work remains to be done.  Most importantly, while our results
relate GAMP to standard optimization methods, these do not guarantee the
algorithm's convergence.
As mentioned in the Introduction,
for general $\Abf$, it is well-known that GAMP methods may diverge \cite{RanSchFle:14-ISIT,Caltagirone:14-ISIT}.
Several recent modifications have been proposed
to improve the stability of GAMP, including
damping \cite{RanSchFle:14-ISIT,Vila:ICASSP:15}.
One potential line of future work is to consider alternates to GAMP
that are based on direct minimization of the energy function.
Some preliminary works in this regard have been presented
in \cite{manoel2014swamp} which proposes a coordinate descent method
and \cite{Rangan:arxiv:15} which uses an ADMM-based method.

GAMP-based methods have also been extended in a wide variety of ways,
such as combining EM with GAMP~\cite{KrzMSSZ:11-arxiv,Vila:TSP:13,Vila:TSP:14,KamRanFU:12-nips,KamRanFU:12-IT},
turbo and hybrid GAMP methods
\cite{SomS:12,RanganFGS:12-ISIT}, applications in
dictionary learning and matrix factorization
\cite{RanganF:12-ISIT,parker2013bilinear,parker2013bilinear2,krzakala2013phase,lesieur2015mmse}, and applications in 
blind deconvolution and self-calibration
\cite{parker2015bilinear}.
Another line of work would be to understand if one can find
free energy and optimization interpretations of these algorithms.
For dictionary learning and matrix factorization
some initial work has appeared in
\cite{kabashima2014phase,krzakala2013phase}.

\section*{Acknowledgements}
The authors would like to thank Ulugbek Kamilov and Vivek K Goyal for
their valuable comments.

\appendices

\section{Approximate Diagonals} \label{sec:approxDiag}

Given a matrix $\Abf \in \R^{m \x n}$ and positive vectors $\dbf_x \in \R^n$ and $\dbf_z$,
consider the positive matrices \eqref{eq:Qxz}.
We analyze the asymptotic behavior of these
matrices under the following assumptions:

\begin{assumption}  \label{as:Qxz}
Consider a sequence of matrices $\Qbf_x$ and
$\Qbf_z$ of the form \eqref{eq:Qxz}, indexed by the dimension $n$
satisfying:
\begin{itemize}
\item[(a)]  The dimension $m$ is a deterministic
function of $n$ with $\lim_{n \arr \infty} m/n = \beta$
for some $\beta > 0$,
\item[(b)]  The positive vectors $\dbf_x$ and $\dbf_z$
are deterministic vectors with
\[
    \limsup_{n \arr \infty} \|\dbf_x\|_\infty < \infty, \quad
    \limsup_{n \arr \infty} \|\dbf_z\|_\infty < \infty.
\]
\item[(c)]  The components of $\Abf$ are
 independent, zero-mean with $\var(A_{ij}) = S_{ij}$
 for some deterministic matrix $\Sbf$  such that
\[
    \limsup_n \max_{i,j}nS_{ij} < \infty.
\]
\end{itemize}
\end{assumption}

\begin{theorem}[\cite{HacLouNaj:07}]  Consider a sequence of matrices $\Qbf_x$ and $\Qbf_z$
in Assumption \ref{as:Qxz}.
Then, for each $n$, there exists positive vectors
$\xibf_x$ and $\xibf_z$ satisfying the
nonlinear equations
\beq \label{eq:QxzFP}
    \mathbf{1}./\xibf_z = \mathbf{1}./\dbf_z+\Sbf\xibf_x, \quad
    \mathbf{1}./\xibf_x = \mathbf{1}./\dbf_x+\Sbf^T\xibf_z,
\eeq
where the vector inverses are componentwise.
Moreover, the vectors
$\xibf_z$ and $\xibf_x$ are asymptotic diagonals of $\Qbf_x$
and $\Qbf_z$ in the following sense:  For any deterministic
sequence of positive vectors $\ubf_x \in \R^n$ and $\ubf_z \in \R^m$,
such that
\[
    \limsup_{n \arr \infty} \|\ubf_x\|_\infty < \infty, \quad
    \limsup_{n \arr \infty} \|\ubf_z\|_\infty < \infty,
\]
the following limits hold almost surely
\beqan
    \lim_{n \arr \infty} \frac{1}{n}\sum_{j=1}^n\left[
        u_{xj}((Q_x)_{jj}-\xi_{xj})\right] &=& 0 \\
    \lim_{n \arr \infty} \frac{1}{m}\sum_{i=1}^m\left[
        u_{zi}((Q_z)_{ii}-\xi_{zi})\right] &=& 0.
\eeqan
\end{theorem}
\begin{proof}  This result is a special case of the results in
\cite{HacLouNaj:07}.
\end{proof}

The result says that, for certain large random matrices $\Abf$,
$\xibf_x$ and $\xibf_z$
are approximate diagonals of the matrices $\Qbf_x$ and $\Qbf_z$,
respectively.  This motivates the following definition
for deterministic $\Abf$.

\begin{definition} \label{def:approxDiag}
Consider matrices $\Qbf_x$ and $\Qbf_z$ of the form \eqref{eq:Qxz}
for some \emph{deterministic} (i.e.\ non-random)
$\Abf$, $\dbf_x$ and $\dbf_z$.  Let $\Sbf = \Abf.\Abf$ be the
componentwise square of $\Abf$.  Then, the unique
positive solutions $\xibf_z$ and $\xibf_x$ to \eqref{eq:QxzFP}
will be called the \emph{approximate diagonals} of $\Qbf_z$ and
$\Qbf_x$, respectively.
\end{definition}

\section{Proof of Theorem \ref{thm:gampMS}} \label{sec:gampMSPf}

To prove \eqref{eq:ziterMS}, observe that
\beqa
    \lefteqn{ \argmin_{\zbf}
    \Bigl[ L(\xbf^t,\zbf,\sbf^{\tm1})
        +\frac{1}{2} \|\zbf-\Abf\xbf^t\|^2_{\taubf_p^t}
        \Bigr] } \nonumber \\
     &\stackrel{(a)}{=}& \argmin_{\zbf}
        \Bigl[ f_z(\zbf) + (\sbf^{\tm1})^T\zbf
        +\frac{1}{2}\|\zbf-\Abf\xbf^t\|^2_{\taubf_p^t}
            \Bigr]       \nonumber \\
        &\stackrel{(b)}{=}& \argmin_{\zbf}
	\Bigl[ f_z(\zbf) + \frac{1}{2}\|\zbf-\pbf^t\|^2_{\taubf_p^t} \Bigr]
	\stackrel{(c)}{=} \zbf^t,
        \nonumber
\eeqa
where (a) follows from substituting \eqref{eq:optsep}
and \eqref{eq:Lag} into \eqref{eq:ziterMS} and eliminating the terms
that do not depend on $\zbf$;
(b) follows from the definition of $\pbf^t$
in line \ref{line:phat};
and (c) follows from the definition of $\zbf^t$ in line \ref{line:zhatMS}.
This proves \eqref{eq:ziterMS}.  The update \eqref{eq:xiterMS}
can be proven similarly.
To prove \eqref{eq:siterMS}, observe that
\[
    \sbf^t \stackrel{(a)}{=} (\zbf^t-\pbf^t)./\taubf_p^t
    \stackrel{(b)}{=} \sbf^{\tm1} + (\zbf^t-\Abf\xbf^t)./\taubf_p^t
\]
where (a) follows from the update of $\sbf^t$ in
line \ref{line:shat} in Algorithm~\ref{algo:gamp}
(recall that the division is componentwise); and
(b) follows from the update for $\pbf^t$ in line \ref{line:phat}.
We have thus proven the equivalence of the max-sum GAMP
algorithm with the Lagrangian updates \eqref{eq:lagMS}.

Now consider any fixed point $(\zbfhat,\xbfhat,\sbf)$ of the max-sum GAMP algorithm.
A fixed point of \eqref{eq:siterMS}, requires that
\beq \label{eq:zxfixMS}
    \zbfhat = \Abf\xbfhat
\eeq
so the fixed point satisfies the constraint of the optimization
\eqref{eq:xzOpt}.
Now, using \eqref{eq:zxfixMS} and
the fact that $\zbfhat$ is the minima of \eqref{eq:ziterMS}, we have that
\[
    \gradient_{\zbf} 
    L(\xbfhat,\zbfhat,\sbf) = \mathbf{0}.
\]
Similarly, since $\xbf$ is the minima
of \eqref{eq:xiterMS}, we have that
\[
    \gradient_{\zbf}
    L(\xbfhat,\zbfhat,\sbf) = \mathbf{0}.
\]
Thus, the fixed point $(\xbfhat,\zbfhat,\sbf)$ is a critical point of the Lagrangian
\eqref{eq:Lag}.

Finally, consider the quadratic terms $(\taubf_x,\taubf_r,\taubf_s)$ at the fixed point.
From the updates of $\taubf_x$ and $\taubf_r$ in Algorithm~\ref{algo:gamp} [see also \eqref{eq:proxderiv}]
and the definition of $\dbf_x$ in \eqref{eq:dxz}, we obtain
\beq
    \mathbf{1}./\taubf_x = \dbf_x + \mathbf{1}./\taubf_r
    = \dbf_x + \Sbf^T\taubf_s.
\eeq
Similarly, the updates of $\taubf_s$ and $\taubf_p$ show that
\beq
    \mathbf{1}./\taubf_s = \mathbf{1}./\dbf_z + \taubf_p =
    \mathbf{1}./\dbf_z + \Sbf\taubf_x.
\eeq

Then, according to Definition~\ref{def:approxDiag}, $\taubf_x$ and $\taubf_s$ are the approximate diagonals of $\Qbf_x$ and $\Qbf_z$ in \eqref{eq:Qxz}, respectively.

\section{Proof of Theorem \ref{thm:gampSP}} \label{sec:gampSPPf}

We prove this theorem in two parts.  First we show that the sum-product GAMP
updates are equivalent to \eqref{eq:lagInSP} and \eqref{eq:lagOutSP}.
Then we show that any fixed points of these updates are critical points
of the constrained optimization \eqref{eq:JBFEopt}.

\subsection{Equivalence of the Updates}

We begin by proving \eqref{eq:lagInSP}.
Define $b_x^{\tp1}$ as the solution to the minimization \eqref{eq:lagInSP}.
So, we must show that this solution is given by the equation for $b_x^{\tp1}(\xbf)$
in \eqref{eq:bxzSP}.  We use induction: Suppose that $b_x^{\tp1}$ in
\eqref{eq:bxzSP} is the solution to \eqref{eq:lagInSP} for some $t$.
We will then show that it is the solution for $t+1$.

First, combining the induction hypothesis that $b_x^{\tp1}$ is given in \eqref{eq:bxzSP}
with lines~\ref{line:xhatSP} and \ref{line:tauxSP} of Algorithm~\ref{algo:gamp}, we have
\beq \label{eq:xbx}
    \xbf^t = \Exp(\xbf|b_x^t), \quad
    \taubf_x^t = \var(\xbf|b_x^t).
\eeq
That is, $\xbf^t$ and $\taubf_x^t$ are the mean and variance vectors of the
density $b_x^t$.
We next simplify the right hand side of \eqref{eq:lagInSP} to remove terms
that do not depend on $b_x$:
\beqa
   \lefteqn{ L_{\rm SP}(b_x,b_z^t,\taubf_p^t,\sbf^t)
    + \frac{1}{2}(\taubf_s^t)^T\Sbf\var(\xbf|b_x) }
    \nonumber \\ & &
    + \frac{1}{2}\left\|\Exp(\xbf|b_x)-\Exp(\xbf|b_x^t)
        \right\|^2_{\taubf_r^t} \nonumber \\
   &\stackrel{(a)}{=}& D(b_x\|e^{-f_x}) -(\sbf^t)^T\Abf\Exp(\xbf|b_x) +  \frac{1}{2}(\taubf_s^t)^T\Sbf\var(\xbf|b_x) \nonumber \\
    & &+ \frac{1}{2}\left\|\Exp(\xbf|b_x)-\Exp(\xbf|b_x^t)
        \right\|^2_{\taubf_r^t} + \mbox{const} \nonumber \\
   &\stackrel{(b)}{=}& D(b_x\|e^{-f_x}) -(\sbf^t)^T\Abf\Exp(\xbf|b_x) +  \left(\frac{1}{2\taubf_r^t}\right)^T\var(\xbf|b_x) \nonumber \\
    & &+ \frac{1}{2}\left\|\Exp(\xbf|b_x)-\xbf^t
        \right\|^2_{\taubf_r^t} + \mbox{const}  \\
   &\stackrel{(c)}{=}& D(b_x\|e^{-f_x}) + \left(\frac{1}{2\taubf_r^t}\right)^T\var(\xbf|b_x) \nonumber \\
    & &+ \frac{1}{2}\left\|\Exp(\xbf|b_x)-\rbf^t\right\|^2_{\taubf_r^t} + \mbox{const} \nonumber \\
   &\stackrel{(d)}{=}& D(b_x\|e^{-f_x})
    + \frac{1}{2}\Exp\left(\|\xbf-\rbf^t\|^2_{\taubf_r^t} \,\middle|\, b_x \right)
   + \mbox{const}, \label{eq:LSPbx}
\eeqa
where in all the steps
``const" denotes any terms that do not depend on $b_x$,
and (a) follows from
the definition of the Lagrangian \eqref{eq:LSP} and the objective
function \eqref{eq:JSP};
(b) follows from \eqref{eq:xbx} and the fact that
$\taubf_r^t = \mathbf{1}./(\Sbf^T\taubf_s^t)$ in
line~\ref{line:taur} of Algorithm~\ref{algo:gamp};
(c) follows from the definition of $\rbf^t$ in line~\ref{line:rhat}; and
finally (d) follows from the simplification:
\beqa
   \lefteqn{ \left(\mathbf{1}./\taubf_r^t\right)^T\var(\xbf|b_x)
    + \left\|\Exp(\xbf|b_x)-\rbf^t
        \right\|^2_{\taubf_r^t} } \nonumber \\
    &=& \sum_{j=1}^n \left[ \frac{1}{\tau_{r_j^t}} \Big(
    \var(x_j|b_{x_j}) +
    (\Exp(x_j|b_{x_j})-r_j^t)^2 \Big)  \right] \nonumber \\
    &=& \sum_{j=1}^n \left[
    \frac{1}{\tau_{r_j}^t}\Big( \Exp(x_j^2|b_{x_j}) - 2r_j^t\Exp(x_j|b_{x_j})
    \Big) \right] + \mbox{const} \nonumber  \\
    &=& \Exp\left(\|\xbf-\rbf^t\|^2_{\taubf_r^t} \,\middle|\, b_x\right)
    + \mbox{const}. \nonumber
\eeqa
Substituting \eqref{eq:LSPbx} into \eqref{eq:lagInSP},
and using the definition of $p(\xbf|\rbf,\taubf_r)$ in \eqref{eq:pxzrp},
\beqa
    \lefteqn{ b_x^{\tp1} = \argmin_{b_x}
    D(b_x\|e^{-f_x})     + \frac{1}{2}\Exp\left(\|\xbf-\rbf^t\|^2_{\taubf_r^t} \,\middle|\, b_x \right) } \nonumber \\
    &=& \argmin_{b_x} -H(b_x) + \Exp\left(f_x(\xbf)
        + \frac{1}{2}\|\xbf-\rbf^t\|^2_{\taubf_r^t} \,\middle|\, b_x \right) \nonumber \\
    &=& \argmin_{b_x} -H(b_x) - \Exp\left(\log p(\xbf|\rbf^t,\taubf_r^t) \mid b_x \right) \nonumber \\
    &=& \argmin_{b_x} D\left(b_x \,\|\, p(\cdot|\rbf^t,\taubf_r^t) \right),
\eeqa
which proves that $b_x^{\tp1}$ satisfies \eqref{eq:bxzSP}.

Similarly, one can show that the solution $b_z^t$ in \eqref{eq:bzSP} is given
by \eqref{eq:bxzSP}.  In addition, $\zbf^t$ and $\taubf_z^t$ in
lines~\ref{line:zhatSP} and \ref{line:tauzSP}  of Algorithm~\ref{algo:gamp} are the mean and variances of the
estimated densities,
\beq \label{eq:zbz}
    \zbf^t = \Exp(\zbf|b_z^t), \quad \taubf_z^t = \var(\zbf|b_z^t).
\eeq
Equation \eqref{eq:taupSP} follows directly from line~\ref{line:taup} and
\eqref{eq:xbx}.  Also, combining lines~\ref{line:phat} and \ref{line:shat},
we obtain \eqref{eq:siterSP}.

Finally to prove \eqref{eq:tausiterSP}, we take the derivatives
\beqa
    \lefteqn{ 
    \gradient_{\taubf_p}
    L_{\rm SP}(b_x^t,b_z^t,\taubf_p^t,\sbf^t) } \nonumber \\
    &\stackrel{(a)}{=} & 
    \gradient_{\taubf_p}
    H_{\rm gauss}(b_z^t,\taubf_p^t) 
    \stackrel{(b)}{=} \frac{1}{2}\Bigl[ \mathbf{1}./\taubf_p^t -
        \taubf_z^t./(\taubf_p^t.\taubf_p^t) \Bigr] \nonumber \\
    &\stackrel{(c)}{=} & \frac{1}{2}\taubf_s^t \nonumber ,
\eeqa
where (a) follows from removing the terms in \eqref{eq:LSP} that
do not depend on  $\taubf_p$;
(b) can be verified by simply taking the derivative of $H_{\rm gauss}$
in \eqref{eq:Hgauss}
with respect to each component $\tau_{p_i}$; and
(c) follows from the definition of $\taubf_s^t$ in line \ref{line:taus}
of Algorithm \ref{algo:gamp}.  This proves \eqref{eq:tausiterSP},
and we have established that the sum-product GAMP updates are equivalent
to \eqref{eq:lagInSP} and \eqref{eq:lagOutSP}.

\subsection{Characterization of the Fixed Points}

First by substituting the constraint $\taubf_p=\Sbf\var(\xbf|b_x)$, we
can rewrite the optimization \eqref{eq:JBFEopt} as
\beq \label{eq:JBFEopt2}
\begin{aligned}
    & \min_{b_x,b_z} & & J_{\rm SP}(b_x,b_z,\Sbf\var(\xbf|b_x))  \\
    & \text{s.t.} & &   \Exp(\zbf|b_z) = \Abf\Exp(\xbf|b_x).
\end{aligned}
\eeq
Corresponding to this optimization, define the Lagrangian
\beqa
    \widetilde{L}_{\rm SP}(b_x,b_z,\sbf) & = & J_{\rm SP}(b_x,b_z,\Sbf\var(\xbf|b_x)) \nonumber \\
    & & + \sbf^T\left(\Exp(\zbf|b_x) - \Abf\Exp(\xbf|b_x) \right), \label{eq:LtSP}
\eeqa
where $\sbf$ are the dual parameters.
Now, let $(\bhat_x,\bhat_z,\sbf)$ be any fixed points of the updates \eqref{eq:lagInSP}
and \eqref{eq:lagOutSP}.  To show that $(\bhat_x,\bhat_z)$
are critical points of the optimization
\eqref{eq:JBFEopt2}, we need to show that they satisfy the constraint
 $\Exp(\zbf|\bhat_z) = \Abf\Exp(\xbf|\bhat_x)$ and that $(\bhat_x,\bhat_z)$
 are stationary points of the Lagrangian $\widetilde{L}_{\rm SP}(b_x,b_z,\sbf)$.

From \eqref{eq:siterSP}, we have that, at any fixed point $(\bhat_x,\bhat_z)$
\beq \label{eq:linConb}
    \Exp(\zbf|\bhat_z) = \Abf\Exp(\xbf|\bhat_x),
\eeq
and so the linear constraint is satisfied.

To show that $(\bhat_x,\bhat_z)$ are stationary points of the Lagrangian,
we introduce the following notation:
suppose that $V(b)$ is a scalar-valued or vector-valued
functional of a density $b(\ubf)$, and that $\Delta b(\ubf)$
is a perturbation direction of that density.
That is, $\Delta b(\ubf)$ is in the tangent plane of the set of densities, so that  
$\int \Delta b (\ubf) \dif\ubf=0$ and $\Delta b (\ubf) =0$ when $b_0 (\ubf)=0$.
We denote the differential of the functional $V(b)$ in
the direction $\Delta b$ at the point $b=b_0$ by
\[
    \left. \frac{\partial V(b)}{\partial b} \right|_{b=b_0}
    \Cdot \Delta b =
        \lim_{\epsilon \arr 0} \frac{1}{\epsilon} \left[
        V\big( b_0 + \epsilon \Delta b\big) - V(b_0) \right],
\]
which is defined when the limit exists.
See \cite{gelfand2000calculus} for a complete treatment of differentials
of functionals.
Using this notation, we need to show that
\begin{subequations} \label{eq:LSPdbxz}
\beqa
    && \left. \frac{\partial}{\partial b_x} \widetilde{L}_{\rm SP}(b_x,\bhat_z,\sbf)
    \right|_{b_x=\bhat_x} \Cdot \Delta b_x = 0, \label{eq:LSPdbx} \\
    && \left. \frac{\partial}{\partial b_z} \widetilde{L}_{\rm SP}(\bhat_x,b_z,\sbf) \right|_{b_z=\bhat_z} \Cdot \Delta b_z = 0, \label{eq:LSPdbz}
\eeqa
\end{subequations}
for all perturbation directions $\Delta b_x$ and $\Delta b_z$.

To prove \eqref{eq:LSPdbx},
first note that, for any $\Delta b_x$,
the partial derivative of the augmenting term in  \eqref{eq:lagInSP} is given by
\beqa
    \lefteqn{ \left. \frac{1}{2}
    \frac{\partial}{\partial b_x} \left\| \Exp(\xbf|b_x) - \Exp(\xbf|\bhat_x) \right\|^2_{\taubf_r} \right|_{b_x=\bhat_x} \Cdot \Delta b_x} \nonumber \\
    &=
    \left( \Exp(\xbf|\bhat_x) - \Exp(\xbf|\bhat_x)\right)^T \Diag(\taubf_r)^{-1} \nonumber \\
    & \times \frac{\partial}{\partial b_x}
    \Exp(\xbf|\bhat_x) \Cdot \Delta b_x  = 0. \label{eq:dnormz}
\eeqa
Also, since $\bhat_x$ is a minima of \eqref{eq:lagInSP}, it is a stationary
point of the function.  Hence, for any perturbation direction $\Delta b_x$,
\beqa
 \lefteqn{ \frac{\partial}{\partial b_x} \Big[ L_{\rm SP}(b_x,\bhat_z,\taubf_p,\sbf)
    + \frac{1}{2}(\taubf_s)^T\Sbf\var(\xbf|b_x) } \nonumber \\
    & &
    + \frac{1}{2}\left\|\Exp(\xbf|b_x)-\Exp(\xbf|\bhat_x)
        \right\|^2_{\taubf_r} \Bigr]_{b_x=\bhat_x} \Cdot \Delta b_x = 0 \nonumber \\
  &\stackrel{(a)}{\Longleftrightarrow}& \frac{\partial}{\partial b_x} \Big[ L_{\rm SP}(b_x,\bhat_z,\taubf_p,\sbf) \nonumber \\
    && + \frac{1}{2}(\taubf_s)^T\Sbf\var(\xbf|b_x) \Big]_{b_x=\bhat_x} \Cdot \Delta b_x = 0
    \nonumber \\
  &\stackrel{(b)}{\Longleftrightarrow}& \frac{\partial}{\partial b_x} \Big[ L_{\rm SP}(b_x,\bhat_z,\taubf_p,\sbf) \nonumber \\
    && + \frac{\partial }{\partial \taubf_p} L_{\rm SP}(\bhat_x,\bhat_z,\taubf_p,\sbfhat)^T  \taubf_p(b_x) \Big]_{b_x=\bhat_x} \Cdot \Delta b_x = 0
    \nonumber \\
  &\stackrel{(c)}{\Longleftrightarrow}& \frac{\partial}{\partial b_x} \Big[ L_{\rm SP}(b_x,\bhat_z,\taubf_p,\sbf) \Big]_{b_x=\bhat_x} \Cdot \Delta b_x \nonumber \\
    && + \frac{\partial}{\partial \taubf_p} L_{\rm SP}(\bhat_x,\bhat_z,\taubf_p,\sbfhat)^T \frac{\partial}{\partial b_x} \taubf_p(b_x) \Big|_{b_x=\bhat_x} \Cdot \Delta b_x = 0
    \nonumber \\
  &\stackrel{(d)}{\Longleftrightarrow}& \frac{\partial}{\partial b_x} L_{\rm SP}(b_x,\bhat_z,\taubf_p(b_x),\sbf) \Big|_{b_x=\bhat_x} \Cdot \Delta b_x = 0
    \nonumber \\
  &\stackrel{(e)}{\Longleftrightarrow}& \frac{\partial}{\partial b_x} L_{\rm SP}(b_x,\bhat_z,\Sbf\var(\xbf|b_x),\sbf) \Big|_{b_x=\bhat_x} \Cdot \Delta b_x = 0
    \nonumber \\
  &\stackrel{(f)}{\Longleftrightarrow}& \frac{\partial}{\partial b_x} \widetilde{L}_{\rm SP}(b_x,\bhat_z,\sbf) \Big|_{b_x=\bhat_x} \Cdot \Delta b_x = 0,
\eeqa
where (a) follows from \eqref{eq:dnormz};
(b) follows from the fixed points \eqref{eq:taupSP} and \eqref{eq:tausiterSP} and the clarifying notation $\taubf_p(b_x)=\Sbf\var(\xbf|b_x)$;
(c) follows from straightforward calculus;
(d) follows from the multivariable chain rule;
(e) follows from the definition of $\taubf_p(b_x)$; and
(f) follows from the definition of the modified Lagrangian in
\eqref{eq:LtSP}.  This proves \eqref{eq:LSPdbx}. The proof of \eqref{eq:LSPdbz}
is similar.

\bibliographystyle{IEEEtran}
\bibliography{bibl}

\newcommand{\SortNoop}[1]{}
\begin{thebibliography}{10}
\providecommand{\url}[1]{#1}
\csname url@samestyle\endcsname
\providecommand{\newblock}{\relax}
\providecommand{\bibinfo}[2]{#2}
\providecommand{\BIBentrySTDinterwordspacing}{\spaceskip=0pt\relax}
\providecommand{\BIBentryALTinterwordstretchfactor}{4}
\providecommand{\BIBentryALTinterwordspacing}{\spaceskip=\fontdimen2\font plus
\BIBentryALTinterwordstretchfactor\fontdimen3\font minus
  \fontdimen4\font\relax}
\providecommand{\BIBforeignlanguage}[2]{{%
\expandafter\ifx\csname l@#1\endcsname\relax
\typeout{** WARNING: IEEEtran.bst: No hyphenation pattern has been}%
\typeout{** loaded for the language `#1'. Using the pattern for}%
\typeout{** the default language instead.}%
\else
\language=\csname l@#1\endcsname
\fi
#2}}
\providecommand{\BIBdecl}{\relax}
\BIBdecl

\bibitem{RanSRFC:13-ISIT}
S.~Rangan, P.~Schniter, E.~Riegler, A.~Fletcher, and V.~Cevher, ``Fixed points
  of generalized approximate message passing with arbitrary matrices,'' in
  \emph{Proc.\ ISIT}, Jul. 2013, pp. 664--668.

\bibitem{NelWed:72}
J.~A. Nelder and R.~W.~M. Wedderburn, ``Generalized linear models,'' \emph{J.\
  Royal Stat.\ Soc. Series A}, vol. 135, pp. 370--385, 1972.

\bibitem{McCulNel:89}
P.~McCullagh and J.~A. Nelder, \emph{Generalized Linear Models}, 2nd~ed.\hskip
  1em plus 0.5em minus 0.4em\relax Chapman \& Hall, 1989.

\bibitem{RanganFG:12-IT}
S.~Rangan, A.~Fletcher, and V.~K. Goyal, ``Asymptotic analysis of {MAP}
  estimation via the replica method and applications to compressed sensing,''
  \emph{IEEE Trans. Inform. Theory}, vol.~58, no.~3, pp. 1902--1923, Mar. 2012.

\bibitem{Eldar:Book:12}
Y.~C. Eldar and G.~Kutyniok, \emph{Compressed Sensing: Theory and
  Applications}.\hskip 1em plus 0.5em minus 0.4em\relax New York: Cambridge
  Univ. Press, 2012.

\bibitem{ChamDLL:98}
A.~Chambolle, R.~A. DeVore, N.~Y. Lee, and B.~J. Lucier, ``Nonlinear wavelet
  image processing: Variational problems, compression, and noise removal
  through wavelet shrinkage,'' \emph{IEEE Trans. Image Process.}, vol.~7,
  no.~3, pp. 319--335, Mar. 1998.

\bibitem{DaubechiesDM:04}
I.~Daubechies, M.~Defrise, and C.~D. Mol, ``An iterative thresholding algorithm
  for linear inverse problems with a sparsity constraint,'' \emph{Commun. Pure
  Appl. Math.}, vol.~57, no.~11, pp. 1413--1457, Nov. 2004.

\bibitem{VonUnser:04}
C.~Vonesch and M.~Unser, ``Fast iterative thresholding algorithm for
  wavelet-regularized deconvolution,'' in \emph{Proc. SPIE: Wavelet XII}, San
  Diego, CA, 2012.

\bibitem{WrightNF:09}
S.~J. Wright, R.~D. Nowak, and M.~Figueiredo, ``Sparse reconstruction by
  separable approximation,'' \emph{IEEE Trans. Signal Process.}, vol.~57,
  no.~7, pp. 2479--2493, Jul. 2009.

\bibitem{BeckTeb:09}
A.~Beck and M.~Teboulle, ``A fast iterative shrinkage-thresholding algorithm
  for linear inverse problem,'' \emph{SIAM J.\ Imag.\ Sci.}, vol.~2, no.~1, pp.
  183–--202, 2009.

\bibitem{Nesterov:07}
Y.~E. Nesterov, ``Gradient methods for minimizing composite objective
  function,'' \emph{CORE Report}, 2007.

\bibitem{BioDFig:07}
J.~Bioucas-Dias and M.~Figueiredo, ``A new {TwIST}: Two-step iterative
  shrinkage/thresholding algorithms for image restoration,'' \emph{IEEE Trans.
  Image Process.}, vol.~16, no.~12, pp. 2992 -- 3004, Dec. 2007.

\bibitem{BoydPCPE:09}
S.~Boyd, N.~Parikh, E.~Chu, B.~Peleato, and J.~Eckstein, ``Distributed
  optimization and statistical learning via the alternating direction method of
  multipliers,'' \emph{Found. Trends Mach. Learn.}, vol.~3, pp. 1--122, 2010.

\bibitem{Eckstein:MP:92}
J.~Eckstein and D.~Bertsekas, ``On the {D}ouglas-{R}achford splitting method
  and the proximal point algorithm for maximal monotone operators,''
  \emph{Math. Program.}, vol.~5, pp. 293--318, 1992.

\bibitem{Goldstein:JIS:09}
T.~Goldstein and S.~Osher, ``The split {B}regman method for {L1}-regularized
  problems,'' \emph{SIAM J. Imaging Sciences}, vol.~2, no.~2, 2009.

\bibitem{Zhang:JSC:11}
X.~Zhang, M.~Burger, and S.~Osher, ``A unified primal-dual algorithm framework
  based on {B}regman iteration,'' \emph{SIAM J. Sci. Comput.}, vol.~46, pp.
  20--46, 2011.

\bibitem{Combettes:MMS:05}
P.~L. Combettes and V.~R. Wajs, ``Signal recovery by proximal forward-backward
  splitting,'' \emph{Multiscale Model. Simul.}, vol.~4, pp. 1168--1200, 2005.

\bibitem{Tseng:91}
P.~Tseng, ``Applications of a splitting algorithm to decomposition in convex
  programming and variational inequalities,'' \emph{SIAM J. Control and
  Optimization}, vol.~29, no.~1, pp. 119--138, Jan. 1991.

\bibitem{Zhu:Tech:08}
M.~Zhu and T.~Chan, ``An efficient primal-dual hybrid gradient algorithm for
  total variation image restoration,'' UCLA CAM, Tech. Rep. 08-34, 2008.

\bibitem{Esser:Diss:10}
J.~E. Esser, ``Primal dual algorithms for convex models and applications to
  image restoration, registration and nonlocal inpainting,'' Ph.D.
  dissertation, University of California, Los Angeles, 2010.

\bibitem{Chambolle:JMIV:11}
A.~Chambolle and T.~Pock, ``A first-order primal-dual algorithm for convex
  problems with applications to imaging,'' \emph{J. Math. Imaging Vis.},
  vol.~40, pp. 120--145, 2011.

\bibitem{He:JIS:12}
B.~He and X.~Yuan, ``Convergence analysis of primal-dual algorithms for a
  saddle-point problem: From contraction perspective,'' \emph{SIAM J. Imaging
  Sci.}, vol.~5, no.~1, pp. 119--149, 2012.

\bibitem{DonohoMM:09}
D.~L. Donoho, A.~Maleki, and A.~Montanari, ``Message-passing algorithms for
  compressed sensing,'' \emph{Proc. Nat. Acad. Sci.}, vol. 106, no.~45, pp.
  18\,914--18\,919, Nov. 2009.

\bibitem{DonohoMM:10-ITW1}
------, ``Message passing algorithms for compressed sensing {I}: motivation and
  construction,'' in \emph{Proc.\ Info.\ Theory Workshop}, Jan. 2010.

\bibitem{DonohoMM:10-ITW2}
------, ``Message passing algorithms for compressed sensing {II}: analysis and
  validation,'' in \emph{Proc.\ Info.\ Theory Workshop}, Jan. 2010.

\bibitem{BayatiM:11}
M.~Bayati and A.~Montanari, ``The dynamics of message passing on dense graphs,
  with applications to compressed sensing,'' \emph{IEEE Trans. Inform. Theory},
  vol.~57, no.~2, pp. 764--785, Feb. 2011.

\bibitem{Rangan:10-CISS}
S.~Rangan, ``Estimation with random linear mixing, belief propagation and
  compressed sensing,'' in \emph{Proc. Conf. on Inform. Sci. \& Sys.},
  Princeton, NJ, Mar. 2010, pp. 1--6.

\bibitem{Rangan:11-ISIT}
------, ``Generalized approximate message passing for estimation with random
  linear mixing,'' in \emph{Proc. IEEE Int. Symp. Inform. Theory}, Saint
  Petersburg, Russia, Jul.--Aug. 2011, pp. 2174--2178.

\bibitem{CaireSTV:11}
G.~Caire, S.~Shamai, A.~Tulino, and S.~Verd{\'u}, ``Support recovery in
  compressed sensing: Information-theoretic bounds,'' in \emph{Proc. UCSD
  Workshop Inform. Theory \& Its Applications}, La Jolla, CA, Jan. 2011.

\bibitem{BayLelMon:12}
M.~Bayati, M.~Lelarge, and A.~Montanari, ``Universality in polytope phase
  transitions and iterative algorithms,'' in \emph{Proc.\ ISIT}, Jul. 2012, pp.
  1643 --1647.

\bibitem{cakmak2014samp}
B.~{\c{C}}akmak, O.~Winther, and B.~H. Fleury, ``{S-AMP}: Approximate message
  passing for general matrix ensembles,'' in \emph{Proc.\ IEEE Information
  Theory Workshop (ITW)}, 2014, pp. 192--196.

\bibitem{RanSchFle:14-ISIT}
S.~Rangan, P.~Schniter, and A.~Fletcher, ``On the convergence of approximate
  message passing with arbitrary matrices,'' in \emph{Proc.\ ISIT}, Jul. 2014,
  pp. 236--240.

\bibitem{Caltagirone:14-ISIT}
F.~Caltagirone, L.~Zdeborov{\'a}, and F.~Krzakala, ``On convergence of
  approximate message passing,'' in \emph{Proc.\ ISIT}, Jul. 2014, pp.
  1812--1816.

\bibitem{FletcherRVB:11}
A.~K. Fletcher, S.~Rangan, L.~Varshney, and A.~Bhargava, ``Neural
  reconstruction with approximate message passing {(NeuRAMP)},'' in \emph{Proc.
  Neural Information Process. Syst.}, Granada, Spain, Dec. 2011.

\bibitem{chen2010improved}
S.~Chen, H.~Tong, Z.~Wang, S.~Liu, M.~Li, and B.~Zhang, ``Improved generalized
  belief propagation for vision processing,'' \emph{Mathematical Problems in
  Engineering}, vol. 2011, 2010.

\bibitem{KamilovGR:12}
U.~S. Kamilov, V.~K. Goyal, and S.~Rangan, ``Message-passing de-quantization
  with applications to compressed sensing,'' \emph{IEEE Trans. Signal
  Process.}, vol.~60, no.~12, pp. 6270--6281, Dec. 2012.

\bibitem{vila2013hyperspectral}
J.~Vila, P.~Schniter, and J.~Meola, ``Hyperspectral image unmixing via bilinear
  generalized approximate message passing,'' in \emph{SPIE Defense, Security,
  and Sensing}.\hskip 1em plus 0.5em minus 0.4em\relax International Society
  for Optics and Photonics, 2013, pp. 87\,430Y--87\,430Y.

\bibitem{fletcher2014scalable}
A.~K. Fletcher and S.~Rangan, ``Scalable inference for neuronal connectivity
  from calcium imaging,'' in \emph{Proc.\ Neural Information Processing
  Systems}, 2014, pp. 2843--2851.

\bibitem{Vila:ICASSP:15}
J.~Vila, P.~Schniter, S.~Rangan, F.~Krzakala, and L.~Zdeborov{\'a}, ``Adaptive
  damping and mean removal for the generalized approximate message passing
  algorithm,'' in \emph{Proc.\ IEEE ICASSP}, 2015, to appear.

\bibitem{manoel2014swamp}
A.~Manoel, F.~Krzakala, E.~W. Tramel, and L.~Zdeborov{\'a}, ``Sparse estimation
  with the swept approximated message-passing algorithm,''
  \emph{arXiv:1406.4311}, Jun. 2014.

\bibitem{Rangan:arxiv:15}
S.~Rangan, A.~K. Fletcher, P.~Schniter, and U.~S. Kamilov, ``Inference for
  generalized linear models via alternating directions and {B}ethe free energy
  minimization,'' \emph{arXiv:1501.01797}, Jan. 2015.

\bibitem{JavMon:12-arXiv}
A.~Javanmard and A.~Montanari, ``State evolution for general approximate
  message passing algorithms, with applications to spatial coupling,''
  \emph{arXiv:1211.5164 [math.PR].}, Nov. 2012.

\bibitem{Krzakala:14-ISITbethe}
F.~Krzakala, A.~Manoel, E.~W. Tramel, and L.~Zdeborov{\'a}, ``Variational free
  energies for compressed sensing,'' in \emph{Proc.\ ISIT}, Jul. 2014, pp.
  1499--1503.

\bibitem{kabashima2014phase}
Y.~Kabashima, F.~Krzakala, M.~M{\'e}zard, A.~Sakata, and L.~Zdeborov{\'a},
  ``Phase transitions and sample complexity in {B}ayes-optimal matrix
  factorization,'' \emph{arXiv:1402.1298}, 2014.

\bibitem{WainwrightJ:08}
M.~J. Wainwright and M.~I. Jordan, ``Graphical models, exponential families,
  and variational inference,'' \emph{Found. Trends Mach. Learn.}, vol.~1, 2008.

\bibitem{BoutrosC:02}
J.~Boutros and G.~Caire, ``Iterative multiuser joint decoding: Unified
  framework and asymptotic analysis,'' \emph{IEEE Trans. Inform. Theory},
  vol.~48, no.~7, pp. 1772--1793, Jul. 2002.

\bibitem{TanakaO:05}
T.~Tanaka and M.~Okada, ``Approximate belief propagation, density evolution,
  and neurodynamics for {CDMA} multiuser detection,'' \emph{IEEE Trans. Inform.
  Theory}, vol.~51, no.~2, pp. 700--706, Feb. 2005.

\bibitem{GuoW:06}
D.~Guo and C.-C. Wang, ``Asymptotic mean-square optimality of belief
  propagation for sparse linear systems,'' in \emph{Proc. IEEE Inform. Theory
  Workshop}, Chengdu, China, Oct. 2006, pp. 194--198.

\bibitem{Montanari:12-bookChap}
A.~Montanari, ``Graphical model concepts in compressed sensing,'' in
  \emph{Compressed Sensing: Theory and Applications}, Y.~C. Eldar and
  G.~Kutyniok, Eds.\hskip 1em plus 0.5em minus 0.4em\relax Cambridge Univ.
  Press, Jun. 2012, pp. 394--438.

\bibitem{Minka:01}
T.~P. Minka, ``A family of algorithms for approximate {B}ayesian inference,''
  Ph.D. dissertation, Massachusetts Institute of Technology, Cambridge, MA,
  2001.

\bibitem{Seeger:08}
M.~Seeger, ``{B}ayesian inference and optimal design for the sparse linear
  model,'' \emph{J. Machine Learning Research}, vol.~9, pp. 759--813, Sep.
  2008.

\bibitem{Tibshirani:96}
R.~Tibshirani, ``Regression shrinkage and selection via the lasso,'' \emph{J.
  Royal Stat. Soc., Ser. B}, vol.~58, no.~1, pp. 267--288, 1996.

\bibitem{ParikhB:13}
N.~Parikh and S.~Boyd, ``Proximal algorithms,'' \emph{Found. Trends Optimiz.},
  vol.~3, no.~1, pp. 123--231, 2013.

\bibitem{DonohoJM:11arXiv}
D.~Donoho, I.~Johnstone, and A.~Montanari, ``Accurate prediction of phase
  transitions in compressed sensing via a connection to minimax denoising,''
  arXiv:1111.1041v1 [cs.IT]., Nov. 2011.

\bibitem{WeissFree:01}
Y.~Weiss and W.~T. Freeman, ``{On the optimality of solutions of the
  max-product belief-propagation algorithm in arbitrary graphs},'' \emph{IEEE
  Trans. Inform. Theory}, vol.~47, no.~2, pp. 736--744, Feb. 2001.

\bibitem{YedidiaFW:03}
J.~S. Yedidia, W.~T. Freeman, and Y.~Weiss, ``Understanding belief propagation
  and its generalizations,'' in \emph{Exploring Artificial Intelligence in the
  New Millennium}.\hskip 1em plus 0.5em minus 0.4em\relax San Francisco, CA:
  Morgan Kaufmann Publishers, 2003, pp. 239--269.

\bibitem{Heskes:03}
T.~Heskes, ``{Stable fixed points of loopy belief propagation are minima of the
  Bethe free energy},'' in \emph{Proc. Neural Information Process. Syst.},
  Vancouver, Canada, Dec. 2003.

\bibitem{KrzMSSZ:11-arxiv}
F.~Krzakala, M.~M{\'e}zard, F.~Sausset, Y.~Sun, and L.~Zdeborov\'a,
  ``Statistical physics-based reconstruction in compressed sensing,''
  \emph{arXiv:1109.4424}, Sep. 2011.

\bibitem{Vila:TSP:13}
J.~P. Vila and P.~Schniter, ``Expectation-maximization {G}aussian-mixture
  approximate message passing,'' \emph{IEEE Trans.\ Signal Processing},
  vol.~61, no.~19, pp. 4658--4672, Oct. 2013.

\bibitem{Vila:TSP:14}
------, ``An empirical-{Bayes} approach to recovering linearly constrained
  non-negative sparse signals,'' \emph{IEEE Trans.\ Signal Processing},
  vol.~62, no.~18, pp. 4689--4703, Sep. 2014, (see also
  \emph{arXiv:1310.2806}).

\bibitem{KamRanFU:12-nips}
U.~S. Kamilov, S.~Rangan, A.~K. Fletcher, and M.~Unser, ``Approximate message
  passing with consistent parameter estimation and applications to sparse
  learning,'' in \emph{Proc.\ NIPS}, Lake Tahoe, NV, Dec. 2012.

\bibitem{KamRanFU:12-IT}
------, ``Approximate message passing with consistent parameter estimation and
  applications to sparse learning,'' \emph{IEEE Trans. Info. Theory}, vol.~60,
  no.~5, pp. 2969 -- 2985, Apr. 2014.

\bibitem{SomS:12}
S.~Som and P.~Schniter, ``Compressive imaging using approximate message passing
  and a {M}arkov-tree prior,'' \emph{IEEE Trans. Signal Process.}, vol.~60,
  no.~7, pp. 3439--3448, Jul. 2012.

\bibitem{RanganFGS:12-ISIT}
S.~Rangan, A.~K. Fletcher, V.~K. Goyal, and P.~Schniter, ``Hybrid generalized
  approximation message passing with applications to structured sparsity,'' in
  \emph{Proc. IEEE Int. Symp. Inform. Theory}, Cambridge, MA, Jul. 2012, pp.
  1241--1245.

\bibitem{RanganF:12-ISIT}
S.~Rangan and A.~K. Fletcher, ``Iterative estimation of constrained rank-one
  matrices in noise,'' in \emph{Proc. IEEE Int. Symp. Inform. Theory},
  Cambridge, MA, Jul. 2012.

\bibitem{parker2013bilinear}
J.~Parker, P.~Schniter, and V.~Cevher, ``Bilinear generalized approximate
  message passing---{Part I: D}erivation,'' \emph{IEEE Trans.\ Signal
  Processing}, vol.~62, no.~22, pp. 5839 -- 5853, 2013.

\bibitem{parker2013bilinear2}
------, ``Bilinear generalized approximate message passing---{Part II:
  A}pplications,'' \emph{IEEE Trans.\ Signal Processing}, vol.~62, no.~22, pp.
  5854--5867, 2013.

\bibitem{krzakala2013phase}
F.~Krzakala, M.~M{\'e}zard, and L.~Zdeborov{\'a}, ``Phase diagram and
  approximate message passing for blind calibration and dictionary learning,''
  in \emph{Proc. IEEE ISIT}, 2013, pp. 659--663.

\bibitem{lesieur2015mmse}
T.~Lesieur, F.~Krzakala, and L.~Zdeborov{\'a}, ``{MMSE} of probabilistic
  low-rank matrix estimation: Universality with respect to the output
  channel,'' \emph{arXiv:1507.03857}, 2015.

\bibitem{parker2015bilinear}
J.~Parker, Y.~Shou, and P.~Schniter, ``Parametric bilinear generalized
  approximate message passing,'' \emph{arXiv:1508.07575}, 2015.

\bibitem{HacLouNaj:07}
W.~Hachem, P.~Loubaton, and J.~Najim, ``Deterministic equivalents for certain
  functionals of large random matrices,'' \emph{Ann. Applied Probability},
  vol.~17, no.~3, pp. 875--930, Jun. 2007.

\bibitem{gelfand2000calculus}
I.~M. Gelfand and S.~V. Fomin, \emph{Calculus of Variations}.\hskip 1em plus
  0.5em minus 0.4em\relax Courier Corporation, 2000.

\end{thebibliography}

\end{document}